\documentclass[a4paper,fleqn,11pt]{article}

\usepackage{amsmath}
\usepackage{amsthm}
\usepackage{amssymb}
\usepackage{mathtools}
\usepackage{accents}
\usepackage{float}

\usepackage[a4paper,top=3cm, bottom=3cm, left=3cm, right=3cm]{geometry}
\usepackage[shortlabels,inline]{enumitem}
\usepackage[pdftex]{graphicx}
\usepackage[pdftex,ocgcolorlinks,pagebackref=false]{hyperref}
\usepackage{authblk}

\setlist[enumerate,1]{label={(\roman*)}}

\usepackage[capitalise,noabbrev]{cleveref}

\Crefname{property}{Property}{Properties}

\theoremstyle{plain}
\newtheorem{theorem}{Theorem}[section]
\newtheorem{lemma}[theorem]{Lemma}
\newtheorem{proposition}[theorem]{Proposition}
\newtheorem{corollary}[theorem]{Corollary}
\newtheorem{remark}[theorem]{Remark}
\newtheorem{example}[theorem]{Example}

\theoremstyle{definition}
\newtheorem{definition}[theorem]{Definition}

\usepackage{pgfplots}
\pgfplotsset{compat=1.13}
\usepgfplotslibrary{fillbetween}
\usetikzlibrary{patterns}
\pgfplotsset{
  spectrumplotstyle/.style={
    grid=both, 
    axis lines=middle,
    xmin=0,
    xmax=1.1,
    xtick={1},
    extra x ticks={0},
    xlabel=$\alpha$,
    xlabel style={anchor=north},
    ylabel style={anchor=south east},
    width=0.6\textwidth,
    height=0.45\textwidth,
  },
  rateplotstyle/.style={
    grid=both, 
    axis lines=middle,
    xmin=0,
    xmax=0.7,
    extra x ticks={0},
    xlabel=$r$,
    xlabel style={anchor=north},
    ylabel=$R$,
    ylabel style={anchor=south east},
    width=0.6\textwidth, 
    height=0.45\textwidth, 
  },
  plotline/.style={
    black,
    semithick,
    mark=none,
  },
  plotlinewithmarks/.style={
    black,
    semithick,
    mark=*,
    mark size=1 pt,
  },
  exampleline/.style={
    black,
    semithick,
    dashed,
    mark=none,
  },
  partofspectrumregion/.style={
    pattern=north west lines,
    pattern color=black!50!white
  },
  notpartofspectrumregion/.style={
    black!50!white,
    fill,
    opacity=0.5
  },
  achievableregion/.style={
    black!50!white,
    fill,
    opacity=0.5
  },
  notachievableregion/.style={
    pattern=north west lines,
    pattern color=black!50!white
  },
}
\tikzset{
  examplepoint/.style={
    black,
    fill,
    circle,
    scale=0.3,
  },
}

\newcommand{\setbuild}[2]{\left\{#1\middle|#2\right\}}

\newcommand{\typeclass}[2]{T^{#1}_{#2}}
\newcommand{\typeclassprojection}[2]{\Pi^{#1}_{#2}}

\newcommand{\ed}{\mathop{}\!\mathrm{d}}
\newcommand{\norm}[2][]{\left\|#2\right\|_{#1}}

\newcommand{\ket}[1]{\left|#1\right\rangle}
\newcommand{\bra}[1]{\left\langle #1\right|}
\newcommand{\ketbra}[2]{\left|#1\middle\rangle\!\middle\langle#2\right|}
\newcommand{\braket}[2]{\left\langle#1\middle|#2\right\rangle}
\newcommand{\vectorstate}[1]{\ketbra{#1}{#1}}

\newcommand{\CW}{\textnormal{CW}}
\newcommand{\EPR}{\textnormal{EPR}}
\newcommand{\GHZ}{\textnormal{GHZ}}
\newcommand{\W}{\textnormal{W}}

\DeclareMathOperator{\boundeds}{\mathcal{B}}
\DeclareMathOperator{\Tr}{Tr}

\newcommand{\entropy}{H}
\newcommand{\mutualinformation}{I}
\newcommand{\relativeentropy}[3][]{\mathop{D_{#1}}\mathopen{}\left(#2\middle\|#3\right)\mathclose{}}
\newcommand{\binaryrelativeentropy}[3][]{\mathop{d_{#1}}\mathopen{}\left(#2\middle\|#3\right)\mathclose{}}
\newcommand{\thesandwicheddivergence}[1][]{\tilde{D}_{#1}}
\newcommand{\sandwiched}[3][]{\mathop{\thesandwicheddivergence[#1]}\mathopen{}\left(#2\middle\|#3\right)\mathclose{}}

\newcommand{\distributions}[1][]{\mathcal{P}_{#1}}
\DeclareMathOperator{\states}{\mathcal{S}}
\newcommand{\unittensor}[1]{\langle{#1}\rangle}

\newcommand{\reals}{\mathbb{R}}
\newcommand{\complexes}{\mathbb{C}}
\newcommand{\naturals}{\mathbb{N}}
\newcommand{\integers}{\mathbb{Z}}
\newcommand{\nonnegativereals}{\mathbb{R}_{\ge 0}}
\newcommand{\positivereals}{\mathbb{R}_{>0}}

\newcommand{\majorizes}{\succcurlyeq}

\newcommand{\polyle}{\stackrel{.}{\le}}
\newcommand{\polyge}{\stackrel{.}{\ge}}
\newcommand{\polyeq}{\stackrel{.}{=}}

\DeclareMathOperator{\supp}{supp}
\DeclareMathOperator{\argmin}{argmin}

\title{Explicit error bounds for entanglement transformations between sparse multipartite states}

\author[1,2]{D\'avid Bug\'ar}
\author[1,2]{P\'eter Vrana}
\affil[1]{Department of Algebra and Geometry, Institute of Mathematics, Budapest University of Technology and Economics, M\H uegyetem~rkp. 3., H-1111 Budapest, Hungary.}
\affil[2]{MTA-BME Lend\"ulet Quantum Information Theory Research Group, M\H uegyetem~rkp. 3., H-1111 Budapest, Hungary}

\begin{document}

\maketitle

\begin{abstract}
The trade-off relation between the rate and the strong converse exponent for probabilistic asymptotic entanglement transformations between pure multipartite states can in principle be characterised in terms of a class of entanglement measures determined implicitly by a set of strong axioms. A nontrivial family of such functionals has recently been constructed, but their previously known characterisations have so far only made it possible to evaluate them in very simple cases. In this paper we derive a new regularised formula for these functionals in terms of a subadditive upper bound, complementing the previously known superadditive lower bound. The upper and lower bounds evaluated on tensor powers differ by a logarithmically bounded term, which provides a bound on the convergence rate. In addition, we find that on states satisfying a certain sparsity constraint, the upper bound is equal to the value of the corresponding additive entanglement measure, therefore the regularisation is not needed for such states, and the evaluation is possible via a single-letter formula. Our results provide explicit bounds on the success probability of transformations by local operations and classical communication and, due to the additivity of the entanglement measures, also on the strong converse exponent for asymptotic transformations.
\end{abstract}

\section{Introduction}

Understanding the possible transformations between entangled states by local operations and classical communication (LOCC) is a major open problem in quantum information theory. In an asymptotic scenario, given a pair of states $\rho$ and $\sigma$, one aims to characterise the achievable rates $R$ such that a transformation of $\rho^{\otimes n}$ into $\sigma^{\otimes Rn+o(n)}$ is possible in the limit $n\to\infty$, under various error criteria. The strongest requirement is to reach the target state exactly, with probability one for all $n$. This can be relaxed in different ways, by allowing either the success probability or the fidelity to the target to be less than one, or even a combination of both \cite{regula2023overcoming}.

More specifically, one may require the probability to converge to one arbitrarily or as $1-2^{-nr+o(n)}$ for a specified (direct) error exponent $r$ or even allow the probability or fidelity to converge to $0$ (strong converse) as $2^{-rn+o(n)}$, or even arbitrarily fast while remaining nonzero for all $n$ (asymptotic SLOCC). To our knowledge, a complete characterisation of the achievable rates in all of these settings is not available even for transformations between bipartite pure states, which is by far the best understood special case. In general one seeks to express the trade-off relations in terms of entropies and other information quantities. For instance, the entanglement properties of a bipartite pure state depend on its normal form $\ket{\psi_P}=\sum_{i=1}^d\sqrt{P(i)}\ket{i}\otimes\ket{i}$, where the Schmidt coefficients $(P(1),\ldots,P(d))$ form a probability vector, unique up to ordering. The maximal transformation rate from $\psi_P$ to $\psi_Q$ is equal to $\frac{\entropy(P)}{\entropy(Q)}$ when a vanishingly small error is allowed, and $\frac{\entropy_0(P)}{\entropy_0(Q)}$ in the asymptotic SLOCC paradigm (see \cref{sec:preliminaries} for the definitions of these and other entropic quantities). In the special case when the target state is a pair of maximally entangled qubits (also known as an EPR pair or Bell state), the deterministic rate is given by the min-entropy \cite{morikoshi2001deterministic}, and the full trade-off curve is known in the direct and converse domains, both in the probabilistic and approximate settings \cite{hayashi2002error}, given in terms of optimised expressions involving the R\'enyi entanglement entropies $\entropy_\alpha(P)$.

The problem becomes much more complicated when the number of subsystems is greater than $2$. In this case, even for transformations with a vanishing error, the rate is not given by a ratio of quantities depending on the initial and the target state. In three of the aforementioned settings a characterisation of the optimal rate is available, but involves entanglement measures that are not explicitly known but are defined through a set of axioms. The first one is asymptotic SLOCC, which is equivalent to the asymptotic restriction problem for tensors. The characterisation was derived by Strassen in the context of tensors, and actually predates the development of entanglement theory \cite{strassen1988asymptotic}. The other settings are the strong converse domain for probabilistic transformations \cite{jensen2019asymptotic} and approximate transformations with asymptotically vanishing error \cite{vrana2022asymptotic}. In this paper we focus on the entanglement measures relevant to the strong converse exponents for probabilistic transformations, which constitute the asymptotic spectrum of LOCC transformations.

Given a natural number $k\ge 2$, the asymptotic spectrum of LOCC transformations \cite{jensen2019asymptotic} is the set $\Delta_k$ of functionals $F$ on $k$-partite unnormalised state vectors that are invariant under local isometries and satisfy
\begin{enumerate}
    \item $F(\sqrt{p}\ket{\psi})=p^\alpha F(\ket{\psi})$,
    \item $F(\unittensor{r})=r$,
    \item\label{it:specMulti} $F(\ket{\psi}\otimes\ket{\varphi})=F(\ket{\psi})F(\ket{\varphi})$,
    \item\label{it:specAddi} $F(\ket{\psi}\oplus\ket{\varphi})=F(\ket{\psi})+F(\ket{\varphi})$,
    \item $F(\ket{\psi})^{1/\alpha}\ge F(\Pi\ket{\psi})^{1/\alpha}+F((I-\Pi)\ket{\psi})^{1/\alpha}$ when $\Pi$ is a local projection
\end{enumerate}
for some (unique) $\alpha\in[0,1]$, where $\unittensor{r}$ is the unit tensor of rank $r$, i.e., the direct sum of $r$ copies of a normalised product vector. The elements of $\Delta_k$ are also called spectral points.

The asymptotic spectrum of LOCC transformations provides a characterisation of exact asymptotic probabilistic LOCC transformations in the following sense.
Let $\ket{\psi}$ and $\ket{\varphi}$ be normalised and $r,R\in\nonnegativereals$. We say that the rate $R$ is achievable with strong converse exponent $r$ (or simply that the pair $(R,r)$ is achievable) if there is a sequence of LOCC channels transforming $\ket{\psi}^{\otimes n}$ into $\ket{\varphi}^{\otimes Rn+o(n)}$ with probability at least $2^{-rn+o(n)}$. By the results of \cite{jensen2019asymptotic}, $(R,r)$ is achievable iff for all $F\in\Delta_k$ the inequality $F(2^{r/2}\ket{\psi})\ge F(\ket{\varphi})^R$ holds. For fixed $r$, the largest achievable $R$ is equal to
\begin{equation}
    R^*(\ket{\psi}\to\ket{\varphi},r)=\inf_{\substack{F\in\Delta_k  \\  F(\ket{\varphi})\neq 1}}\frac{\log F(2^{r/2}\ket{\psi})}{\log F(\ket{\varphi})}.
\end{equation}
We note that the characterisation is valid also for $r=0$, in this case any decreasing exponential is eventually a lower bound on the success probability (but it may still vanish slower than any exponential). The resulting rate is an upper bound on the largest achievable rate with success probability converging to $1$, and in the bipartite case the two rates are equal, while this is not known for $k\ge 3$ subsystems. The set of spectral points with $\alpha=0$ is the asymptotic spectrum of tensors \cite{strassen1988asymptotic}. Restricting the minimum to this subset gives the optimal rate for asymptotic SLOCC transformations.

It is often more convenient to work with the elements of $\Delta_k$ in a logarithmic form, normalised as $E(\varphi):=\frac{\log F(\ket{\varphi})}{1-\alpha}$ (for the unique $\alpha$ above). We will refer to such a transformed functional as a logarithmic spectral point of order $\alpha$. For instance, when $k=2$ these are precisely the R\'enyi entropies of entanglement of orders $\alpha\in[0,1)$. Note that this normalisation excludes any functional with $\alpha=1$, but it can be shown that the only element of $\Delta_k$ with $\alpha=1$ is the squared norm.

A continuous family of explicit elements of $\Delta_k$ with $k\ge 3$ is known. The construction given in \cite{vrana2023family} is given in terms of the large deviation rate function for a multipartite generalisation of the empirical Young diagram measurement \cite{alicki1988symmetry,keyl2001estimating} (see \cref{sec:preliminaries} for details). Unfortunately, they do not seem to be easily computable apart from the simplest special cases, despite the fact that the rate function can be given by a single-letter formula \cite{franks2020minimal,botero2021large}. Arguably the simplest characterisation is that the functionals (in logarithmic forms) are regularisations of the superadditive quantities
\begin{equation}
E_{\alpha,\theta}(\psi) = \sup_{\varphi=(A_1\otimes\dots\otimes A_k)\psi}\sum_{j=1}^k\theta(j)\entropy\left(\frac{\Tr_j\vectorstate{\varphi}}{\norm{\varphi}^2}\right)+\frac{\alpha}{1-\alpha}\log\norm{\varphi}^2,
\end{equation}
where the supremum is over local contractions $A_j:\mathcal{H}_j\to\mathcal{H}_j$ when $\psi\in\mathcal{H}_1\otimes\dots\otimes\mathcal{H}_k$. While from a computational point of view this characterisation is unsatisfactory due to the presence of the regularisation (which, however, is known to be unnecessary when $\alpha=0$ \cite{christandl2023universal}), it can still be useful as a sequence of lower bounds converging to $E^{\alpha,\theta}$, especially as it is given by a supremum, therefore every feasible point provides a lower bound.

In this paper we derive a new characterisation as the regularisation of a \emph{subadditive} quantity, which therefore provides upper bounds converging to $E^{\alpha,\theta}$. For $\alpha\in(0,1)$ and convex weights $\theta(1),\dots,\theta(k)$ we define the functional $\rho^{\alpha,\theta}$ by
\begin{equation}\label{eq:rhodefintro}
\rho^{\alpha,\theta}(\psi)=\min_{\mathcal{M}_1,\dots,\mathcal{M}_k}\entropy_{\alpha,\theta}((\mathcal{M}_1\otimes\dots\otimes\mathcal{M}_k)(\vectorstate{\psi})),
\end{equation}
where the minimum is over local von Neumann measurement channels $\mathcal{M}_j:\boundeds(\mathcal{H}_j)\to\boundeds(\mathcal{H}_j)$ (with the diagonal output of $\mathcal{M}_1\otimes\dots\otimes\mathcal{M}_k$ regarded as the joint distribution on the product of some index sets $\mathcal{X}_1\times\dots\times\mathcal{X}_k$), and
\begin{equation}
\entropy_{\alpha,\theta}(P)=\max_{Q\in\distributions(\supp P)}\left[\sum_{j=1}^k\theta(j)\entropy(Q_j)-\frac{\alpha}{1-\alpha}\relativeentropy{Q}{P}\right],
\end{equation}
where $\distributions(\supp P)$ is the set of probability distributions with a support contained in that of $P$ (i.e., those with $\relativeentropy{Q}{P}<\infty$). We show that the regularisation of $\rho^{\alpha,\theta}$ is equal to $E^{\alpha,\theta}$ (\cref{thm:rhoregE}). More precisely, we find that $E_{\alpha,\theta}(\psi^{\otimes n})\le nE^{\alpha,\theta}(\psi)\le\rho^{\alpha,\theta}(\psi^{\otimes n})\le E_{\alpha,\theta}(\psi^{\otimes n})+O(\log n)$ (where the implied constant depends on the local dimensions and $\alpha$), putting an upper bound on the convergence rate for the regularisation of both $E_{\alpha,\theta}$ and $\rho^{\alpha,\theta}$.

In addition, we identify a nontrivial set of states that satisfy $\rho^{\alpha,\theta}(\psi)=E^{\alpha,\theta}(\psi)$. This set of states is closed under tensor products and includes well-studied states such as \W{} states and Dicke states as well as their weighted versions. The states in question are characterised by the property that, when expanded in a suitable product basis, have no pair of nonzero coefficients that are adjacent in the sense of Hamming distance, i.e., any two basis elements having nonzero overlap with the state must differ in at least two tensor factors (\cref{def:free}). We will say that these states have free support, in reference to a similar property of tensors considered in \cite{franz2002moment}. This property may be viewed as a sparsity condition since it implies that the number of nonzero coefficients of such a state in (say) $\complexes^d\otimes\dots\otimes\complexes^d$, with $k$ factors is at most $d^{k-1}$, in contrast with the maximum $d^k$. Moreover, we show that any product basis with this property is an optimal choice for the measurement basis in \eqref{eq:rhodefintro}, eliminating the need to minimise over the local bases for such states, and reducing the computation of $E^{\alpha,\theta}(\psi)$ to a convex program (\cref{thm:freeErho}). For general states, the minimisation may not be simple but any particular basis choice provides an upper bound on $E^{\alpha,\theta}$.

The functionals $\rho^{\alpha,\theta}$ are analogous to the upper support functionals $\rho^\theta$ defined for tensors over arbitrary fields in \cite[eq. (2.4)]{strassen1991degeneration} as
\begin{equation}
\rho^\theta(\psi)=\min_B\max_{Q\in\distributions(\supp_B\psi)}\sum_{j=1}^k\theta(j)\entropy(Q_j),
\end{equation}
where the minimisation is over all possible product bases (not necessarily orthonormal ones), and $\supp_B\psi$ denotes the set of indices of the basis elements with nonzero coefficient in the expansion of $\psi$. Over the complex numbers, the support functionals are related in a similar way to the quantum functionals introduced in \cite{christandl2023universal} as $\rho^{\alpha,\theta}$ to $E^{\alpha,\theta}$. However, while the $\alpha\to 0$ limit of $E^{\alpha,\theta}$ is equal to the quantum functional with the same weights $\theta$, we do not know if $\lim_{\alpha\to 0}\rho^{\alpha,\theta}$ is equal to $\rho^\theta$. The reason for the possible difference is that in \eqref{eq:rhodefintro} we minimize over orthonormal bases, while $\rho^\theta$ does not depend on the inner product and allows arbitrary local bases, therefore the minimum is potentially lower. In a similar way, while the results on the regularization of $\rho^{\alpha,\theta}$ and the tensorisation property for states with free support are analogous to the similar properties of the support functionals proved in \cite{christandl2023universal}, there does not seem to be a simple implication in any direction. Nevertheless, we find it useful to think of our functionals $\rho^{\alpha,\theta}$ as R\'enyi generalizations of the support functionals with order parameter $\alpha\in(0,1)$.

The structure of the paper is the following. In \cref{sec:preliminaries} we review some of the properties of the Shannon and R\'enyi entropies and the corresponding divergences, and recall some facts from the representation theory of the symmetric and unitary groups as well as known characterizations of the functionals $E^{\alpha,\theta}$. In \cref{sec:rhogeneral} derive the upper bound $E^{\alpha,\theta}\le\rho^{\alpha,\theta}$, starting from a characterisation of $E^{\alpha,\theta}$ in terms of a regularised relative entropy distance and utilizing a data processing inequality for tensor powers of a product measurement channel. Along the way we study the properties of the (classical) entropic quantities $\entropy_{\alpha,\theta}$, and in particular find that they are additive. We also show here that the regularization of $\rho^{\alpha,\theta}$ is equal to $E^{\alpha,\theta}$ for all states. In \cref{sec:rhofree} we specialise to states having a free support and show that on these states the regularisation is not necessary. In \cref{sec:examples} we illustrate our results by evaluating $E^{\alpha,\theta}$ and the corresponding bounds on the trade-off curve between the strong converse exponent and the transformation rate for concrete states and transformations.

\section{Preliminaries}\label{sec:preliminaries}
We begin with basic definitions and results related to classical information theory and quantum Shannon theory, all of which can be found in the books \cite{csiszar2011information,cover2012elements,wilde2013quantum,tomamichel2015quantum}. These are followed by a brief review of some facts from the representation theory of symmetric and unitary groups. For more on these we refer the reader to \cite{fulton1991representation,hayashi2017group}.

We identify measures on a finite set $\mathcal{X}$ with elements of $\nonnegativereals^\mathcal{X}$, and denote the set of probability measures (i.e., measures $P$ with $\norm[1]{P}=\sum_{x\in\mathcal{X}}P(x)=1$) by $\distributions(\mathcal{X})$. The support of a measure $P$ is the subset $\supp P=\setbuild{x\in\mathcal{X}}{P(x)\neq 0}$. If $P\in\distributions(\mathcal{X}_1\times\dots\times\mathcal{X}_k)$ is a distribution on a product set, then we may consider its marginals, e.g.
\begin{equation}
P_1(x_1)=\sum_{x_2,\dots,x_k}P(x_1,\dots,x_k).    
\end{equation}
If $P^{(1)}$ and $P^{(2)}$ are measures on product sets $\mathcal{X}^{(1)}_1\times\dots\times\mathcal{X}^{(1)}_k$ and $\mathcal{X}^{(2)}_1\times\dots\times\mathcal{X}^{(2)}_k$ respectively, then we view their product $P^{(1)}\otimes P^{(2)}$ as a measure on the product set $(\mathcal{X}^{(1)}_1\times\mathcal{X}^{(2)}_1)\times\dots\times(\mathcal{X}^{(1)}_k\times\mathcal{X}^{(2)}_k)$ with $k$ factors. Similarly, the direct sum $P^{(1)}\oplus P^{(2)}$ is considered as a measure on $(\mathcal{X}^{(1)}_1\sqcup\mathcal{X}^{(2)}_1)\times\dots\times(\mathcal{X}^{(1)}_k\sqcup\mathcal{X}^{(2)}_k)$. Note in particular that $(P^{(1)}\otimes P^{(2)})_j=P^{(1)}_j\otimes P^{(2)}_j$ and $(P^{(1)}\oplus P^{(2)})_j=P^{(1)}_j\oplus P^{(2)}_j$ hold for the marginals.

Given a finite set $\mathcal{X}$ with $\lvert\mathcal{X}\vert=d$ and a vector $v\in\nonnegativereals^\mathcal{X}$, the vector $v^\downarrow\in\nonnegativereals^d$ is formed by sorting the entries of $v$ nonincreasingly. We do not consider two such nonincreasing vectors different if they differ only in trailing zeros. We say that $v$ majorizes (or dominates) $w$ and write $v\majorizes w$ if $\sum_{i=1}^m(v^\downarrow)_i\ge\sum_{i=1}^m(w^\downarrow)_i$ holds for all $m$ and $\norm[1]{v}=\norm[1]{w}$. A function $f:\distributions(\mathcal{X})\to\reals$ is Schur concave if $v\majorizes w$ implies $f(v)\le f(w)$. A sufficient condition for this is that $f$ is concave and permutation-invariant.

A nonincreasing vector $\lambda$ is a partition of $n\in\naturals$ if its entries are natural numbers and their sum is equal to $n$. In this case we also write $\lambda\vdash n$ and note that $\lambda/n$ is a probability distribution on $[d]$. The length of a partition is the number of its nonzero entries. 

A probability distribution $Q\in\distributions(\mathcal{X})$ is called an $n$-type if $nQ$ has integer entries. The set of $n$-types will be denoted by $\distributions[n](\mathcal{X})$. A string in $\mathcal{X}^n$ is said to have type $Q$ if for all $x\in\mathcal{X}$ the number of occurrences of $x$ in the string is $nQ(x)$. The set of all such strings is the type class $\typeclass{n}{Q}$. The cardinality of the type class satisfies \cite[Lemma 2.3]{csiszar2011information}
\begin{equation}\label{eq:typeclassbound}
\frac{1}{(n+1)^{\lvert\mathcal{X}\rvert}}2^{n\entropy(Q)}\le\lvert\typeclass{n}{Q}\rvert\le 2^{n\entropy(Q)}.
\end{equation}

A state on a (finite-dimensional) Hilbert space $\mathcal{H}$ is a linear operator $\rho$ on $\mathcal{H}$ such that $\rho\ge 0$ and $\Tr\rho=1$. An orthonormal basis $\{\ket{x}\}_{x\in\mathcal{X}}$ determines a (von Neumann) measurement channel
\begin{equation}
    \mathcal{M}(\rho)=\sum_{x\in\mathcal{X}}\ketbra{x}{x}\rho\ketbra{x}{x}.
\end{equation}
Conversely, such a channel determines an orthonormal basis up to a choice of $\lvert\mathcal{X}\rvert$ phases. This ambiguity will not make a difference in our results, therefore we will use orthonormal bases and measurement channels interchangeably. We will identify states that are diagonal with respect to a preferred basis with probability distributions on the index set.

The orthonormal basis $\{\ket{x}\}_{x\in\mathcal{X}}$ together with an $n$-type $Q\in\distributions[n](\mathcal{X})$ gives rise to the type class projection
\begin{equation}\label{eq:typeclassprojection}
\typeclassprojection{n}{Q}=\sum_{x\in\typeclass{n}{Q}}\ketbra{x}{x}.
\end{equation}
Note that the dependence on the basis is not reflected in the notation. When the quantum system consists of $k$ subsystems, i.e., $\mathcal{H}=\mathcal{H}_1\otimes\dots\otimes\mathcal{H}_k$, and we choose an orthonormal basis in each of them (with index sets $\mathcal{X}_j$), then the family of all the possible tensor products of basis elements is an orthonormal basis of $\mathcal{H}$ indexed by the product set $\mathcal{X}_1\times\dots\times\mathcal{X}_k$. When $Q\in\distributions[n](\mathcal{X}_1\times\dots\times\mathcal{X}_k)$, the joint type class $\typeclass{n}{Q}$ is a subset of the product $\typeclass{n}{Q_1}\times\dots\times\typeclass{n}{Q_k}$, which corresponds to the inequality
\begin{equation}
    \typeclassprojection{n}{Q}\le\typeclassprojection{n}{Q_1}\otimes\dots\otimes\typeclassprojection{n}{Q_k}
\end{equation}
on the level of type class projections.

Let $P\in\distributions(\mathcal{X})$. The R\'enyi entropy of order $\alpha\in(0,1)\cup(1,\infty)$ is defined as $\entropy_\alpha(P)=\frac{1}{1-\alpha}\log\sum_{x\in\supp P}P(x)^\alpha$, where the base of the logarithm is $2$. It is a decreasing function of $\alpha$ and its limit as $\alpha\to 1$ is the Shannon entropy $\entropy(P)=-\sum_{x\in\supp P}P(x)\log P(x)$. The Shannon entropy is concave and permutation-invariant, therefore Schur concave. The Shannon entropy satisfies $0\le\entropy(P)\le\log\lvert\supp P\rvert\le\log\lvert\mathcal{X}\rvert$. Let $(I_i)_{i=1}^m$ be pairwise disjoint subsets of $\mathcal{X}$ with union equal to $\mathcal{X}$ (i.e., a partition of the set), and introduce the distribution $\hat{P}(i)=P(I_i)$ on $[m]$. The Shannon entropy satisfies the following recursion (or chain rule):
\begin{equation}
    \entropy(P)=\entropy(\hat{P})+\sum_{i\in\supp\hat{P}}\hat{P}(i)\entropy(\left.P\right|_{I_i}/\hat{P}(i)).
\end{equation}
For the entropies of a distribution $(p,1-p)$ on the binary alphabet we use the special notations $h_\alpha(p)=\frac{1}{1-\alpha}\log(p^\alpha+(1-p)^\alpha)$ and $h(p)=-p\log p-(1-p)\log(1-p)$.

The extension of the Shannon entropy to quantum states is the von Neumann entropy, defined by $\entropy(\rho)=-\Tr\rho\log\rho$ in the sense of functional calculus for the continuous extension of the function $t\mapsto t\log t$ to $[0,\infty)$. It is also nonnegative and its maximum is $\log\dim\mathcal{H}$. The entropies of a bipartite state $\rho\in\states(\mathcal{H}_A\otimes\mathcal{H}_B)$ and its marginals satisfy the triangle inequality
\begin{equation}
    \entropy(\rho_{AB})\ge\entropy(\rho_A)-\entropy(\rho_B).
\end{equation}

Given a probability distribution $Q$ and a measure $P$ on $\mathcal{X}$, the Kullback--Leibler divergence or relative entropy is defined as
\begin{equation}
\relativeentropy{Q}{P}=\begin{cases}
    \sum_{x\in\supp Q}Q(x)\log\frac{Q(x)}{P(x)} & \text{if $\supp Q\subseteq\supp P$}  \\
    \infty & \text{otherwise.}
\end{cases}
\end{equation}
The relative entropy is jointly convex and satisfies
\begin{equation}
    \relativeentropy{Q}{P}\ge-\log\norm[1]{P},
\end{equation}
with equality iff $Q=\frac{P}{\norm[1]{P}}$. In the binary case we introduce the abbreviation $\binaryrelativeentropy{q}{p}=q\log\frac{q}{p}+(1-q)\log\frac{1-q}{1-p}$.

For measures $P_1$ and $P_2$ on the finite sets $\mathcal{X}_1$ and $\mathcal{X}_2$ respectively, and a distribution $Q\in\distributions(\mathcal{X}_1\times\mathcal{X}_2)$, the Kullback--Leibler divergence satisfies
\begin{equation}\label{eq:KLproductidentity}
\begin{split}
\relativeentropy{Q}{P_1\otimes P_2}
 & = \sum_{x_1,x_2}Q(x_1,x_2)\log\frac{Q(x_1,x_2)}{P_1(x_1)P_2(x_2)}  \\
 & = \sum_{x_1,x_2}Q(x_1,x_2)\log\frac{Q_1(x_1)Q_2(x_2)}{P_1(x_1)P_2(x_2)}\frac{Q(x_1,x_2)}{Q_1(x_1)Q_2(x_2)}  \\
 & = \sum_{x_1,x_2}Q(x_1,x_2)\left(\log\frac{Q_1(x_1)}{P_1(x_1)}+\log\frac{Q_2(x_2)}{P_2(x_2)}+\log\frac{Q(x_1,x_2)}{Q_1(x_1)Q_2(x_2)}\right)  \\
 & = \relativeentropy{Q_1}{P_1}+\relativeentropy{Q_2}{P_2}+\mutualinformation(1:2)_Q,
\end{split}
\end{equation}
where $\mutualinformation(1:2)_Q=\entropy(Q_1)+\entropy(Q_2)-\entropy(Q)\ge 0$ is the mutual information.

There is a chain rule for the relative entropy as well, which we state for two terms for simplicity. Let $P_1$ and $P_2$ be measures on $\mathcal{X}_1$ and $\mathcal{X}_2$, respectively, $Q_1\in\distributions(\mathcal{X}_1)$, $Q_2\in\distributions(\mathcal{X}_2)$, and $q\in[0,1]$. Then
\begin{equation}
    \relativeentropy{qQ_1\oplus(1-q)Q_2}{P_1\oplus P_2}=q\relativeentropy{Q_1}{P_1}+(1-q)\relativeentropy{Q_2}{P_2}-h(q).
\end{equation}

The R\'enyi entropies can be expressed in terms of the Shannon entropy and the Kullback--Leibler divergence via the variational formula \cite{arikan1996inequality,merhav1999shannon,shayevitz2011renyi}
\begin{equation}\label{eq:variationalRenyi}
\entropy_\alpha(P)=\max_{Q\in\supp P}\left[\entropy(Q)-\frac{\alpha}{1-\alpha}\relativeentropy{Q}{P}\right],
\end{equation}
where $\alpha\in(0,1)$.

The (classical) R\'enyi divergence of order $\alpha\in(0,1)$ between the measures $Q$ and $P$ is
\begin{equation}
    \relativeentropy[\alpha]{Q}{P}=\frac{1}{\alpha-1}\log\sum_{x\in\mathcal{X}}Q(x)^\alpha P(x)^{1-\alpha}.
\end{equation}
The R\'enyi divergence has many possible extensions to pairs of positive operators \cite{petz1986quasi,wilde2014strong,muller2013quantum,audenaert2015alpha,mosonyi2022geometric}. We will make use of the minimal or sandwiched R\'enyi divergence \cite{wilde2014strong,muller2013quantum}
\begin{equation}
    \sandwiched[\alpha]{\rho}{\sigma}=\frac{1}{\alpha-1}\log\Tr\left(\sqrt{\rho}\sigma^{\frac{1-\alpha}{\alpha}}\sqrt{\rho}\right)^\alpha.
\end{equation}
When $\alpha\ge\frac{1}{2}$, it satisfies the data processing inequality $\sandwiched[\alpha]{\rho}{\sigma}\ge\sandwiched[\alpha]{T(\rho)}{T(\sigma)}$ for every channel $T$. If the first argument has rank at most $1$, i.e. $\rho=\vectorstate{\psi}$ for some vector $\psi$, then the sandwiched R\'enyi divergence may also be written as
\begin{equation}
    \sandwiched[\alpha]{\vectorstate{\psi}}{\sigma}=\frac{\alpha}{\alpha-1}\log\bra{\psi}\sigma^{\frac{1-\alpha}{\alpha}}\ket{\psi}.
\end{equation}

Following and extending the definition in \cite[Section 3.3]{cover2012elements}, we introduce a notation for comparing nonnegative sequences depending on a natural number $n$ to first order in the exponent: we will write $a_n\polyle b_n$ if for some positive sequence $(r_n)_{n\in\naturals}$ with $\lim_{n\to\infty}\sqrt[n]{r_n}=1$ the inequality $a_n\le r_n b_n$ holds for all $n\in\naturals$. If both $a_n\polyle b_n$ and $a_n\polyge b_n$ then we write $a_n\polyeq b_n$. If the sequence depends on other parameters as well (such as a bounded-length partition $\lambda\vdash n$ or an $n$-type $Q$ over a fixed finite alphabet), we will require the bound to be uniform in the additional parameters and to depend only on $n$. For instance, \eqref{eq:typeclassbound} implies $\lvert\typeclass{n}{Q}\rvert\polyeq 2^{n\entropy(Q)}$. In fact, here and in all the instances below the subexponential factor $r_n$ can be chosen to be a polynomial in $n$. More generally, we will use the same notation to express similar relations between sequences of nonnegative operators.

Given a finite-dimensional Hilbert space $\mathcal{H}$ and a natural number $n$, the Schur--Weyl decomposition writes $\mathcal{H}^{\otimes n}$ as a direct sum of irreducible subspaces for the representation of $U(\mathcal{H})\times S_n$, where the unitary group acts diagonally and the symmetric permutes the tensor factors. The subspaces are labelled by partitions $\lambda\vdash n$ of length at most $d=\dim\mathcal{H}$ and the decomposition has the form
\begin{equation}
    \mathcal{H}^{\otimes n}\simeq\bigoplus_{\lambda\vdash n}\mathbb{S}_\lambda(\mathcal{H})\otimes[\lambda],
\end{equation}
where $\mathbb{S}_\lambda(\mathcal{H})$ is an irreducible representation of $U(\mathcal{H})$ (or zero if $\lambda$ has length greater than $d$) and $[\lambda]$ is an irreducible representation of $S_n$. The number of terms is bounded by $(n+1)^d$ and the dimensions of the appearing spaces satisfy \cite[eqs. (6.16) and (6.21)]{hayashi2017group} (see also \cite{hayashi2002exponents,christandl2006spectra,harrow2005applications})
\begin{equation}\label{eq:unitarydimensionbound}
\dim\mathbb{S}_\lambda(\mathcal{H})\le(n+1)^{d(d-1)/2}
\end{equation}
\begin{equation}\label{eq:symmetricdimensionbound}
\frac{1}{(n+d)^{(d+2)(d-1)/2}}2^{n\entropy(\lambda/n)}\le \dim[\lambda]\le 2^{n\entropy(\lambda/n)}.
\end{equation}
We denote the projection onto the subspace corresponding to $\lambda$ by $P^{\mathcal{H}}_\lambda$ or by $P_\lambda$ if the Hilbert space is clear from the context.

Given an orthonormal basis $\ket{1},\dots,\ket{d}$ of $\mathcal{H}$, let $T\subseteq U(\mathcal{H})$ be the subgroup of diagonal unitaries. We say that a vector in a representation of $U(\mathcal{H})$ is a weight vector with weight $w\in\integers^d$ if it is an eigenvector of every element of $T$ and $\operatorname{diag}(z_1,\dots,z_d)$ acts on it by multiplication with $z_1^{w_1}z_2^{w_2}\cdots z_d^{w_d}$. The dimension of the space of weight-$\mu$ vectors in $\mathbb{S}_\lambda(\mathcal{H})$, when $\lambda$ has at most $d$ parts, is given by the Kostka number $K_{\lambda\mu}$ \cite[\S 15.3]{fulton1991representation}. The Kostka numbers satisfy $K_{\lambda\mu}\neq 0$ iff $\lambda\majorizes\mu$ \cite{lam1977young}. In particular, under the same assumption on the dimension we have
\begin{alignat}{2}
1\le K_{\lambda\mu} & \le\dim(\mathbb{S}_\lambda(\mathcal{H})) &\quad& \text{if $\lambda\majorizes\mu$,}  \label{eq:Kostkabound} \\
K_{\lambda\mu} & = 0 && \text{if $\lambda\not\majorizes\mu$.}  \label{eq:Kostkavanishing}
\end{alignat}

In the multipartite case, a family of elements of the asymptotic spectrum $\Delta_k$ for $k\ge 3$ was constructed in \cite{vrana2023family} (with a possible generalisation considered in \cite{bugar2022interpolating}). These can be viewed as R\'enyi generalisations of the convex combinations of single-site marginal von Neumann entropies, but are not simply functions of the marginals. In the following we briefly review their construction.

 Let $\psi\in\mathcal{H}_1\otimes\dots\otimes\mathcal{H}_k$ be a $k$-partite state and $\overline{\lambda}=(\overline{\lambda}_1,\dots,\overline{\lambda}_k)$ a $k$-tuple of decreasingly ordered finite probability vectors. We define the rate function
\begin{equation}\label{eq:ratefunction}
I_\psi(\overline{\lambda_1},\dots,\overline{\lambda_k})=\lim_{\epsilon\to 0}\lim_{n\to\infty}-\frac{1}{n}\log\sum_{\substack{\lambda_1,\dots,\lambda_k\vdash n  \\  \forall j:\norm[1]{\frac{\lambda_j}{n}-\overline{\lambda_j}}\le\epsilon}}\norm{(P^{\mathcal{H}_1}_{\lambda_1}\otimes\dots\otimes P^{\mathcal{H}_k}_{\lambda_k})\psi^{\otimes n}}^2.
\end{equation}
For $\alpha\in (0,1)$ and convex weights $\theta\in\distributions([k])$ we define
\begin{align}
F^{\alpha,\theta}(\psi) & = 
\begin{cases}
2^{(1-\alpha)E^{\alpha,\theta}(\psi)}\quad &\text{if }\psi \neq 0  \label{eq:upperLOCC}\\
0 &\text{if }\psi = 0
\end{cases}
\intertext{where}
E^{\alpha,\theta}(\psi) & =\sup_{\overline{\lambda_1},\dots,\overline{\lambda_k}}\left[\sum_{j=1}^k\theta(j)\entropy(\overline{\lambda_j})-\frac{\alpha}{1-\alpha}I_\psi(\overline{\lambda_1},\dots,\overline{\lambda_k})\right].  \label{eq:logupperLOCC}
\end{align}
In \cite{vrana2023family} it was shown that for each $\alpha\in(0,1)$ and $\theta\in\distributions([k])$ the functional $F^{\alpha,\theta}$ is an element of the asymptotic spectrum of LOCC transformations.

The parameter $\alpha$ is the order determining the scaling of the functional: $F^{\alpha,\theta}(\sqrt{p}\psi)=p^\alpha F^{\alpha,\theta}(\psi)$. These functionals interpolate between quantum functionals ($\alpha\to 0$) and the $\theta$-weighted average of the von Neumann entropies of marginals for the bipartitions $(\{j\},[k]\setminus \{j\})$, i.e. the entanglement entropy between an elementary subsystem and its complement.

We will make use of two other characterisations, one of which equates $E^{\alpha,\theta}$ with the regularisation of the functional
\begin{equation}\label{eq:lowerLogLOCC}
E_{\alpha,\theta}(\psi) = \sup_{\varphi=(A_1\otimes\dots\otimes A_k)\psi}\sum_{j=1}^k\theta(j)\entropy\left(\frac{\Tr_j\vectorstate{\varphi}}{\norm{\varphi}^2}\right)+\frac{\alpha}{1-\alpha}\log\norm{\varphi}^2,
\end{equation}
where we take the supremum over contractions $A_j:\mathcal{H}_j\to\mathcal{H}_j$ (i.e., $E^{\alpha,\theta}(\psi)=\lim_{n\to\infty}\frac{1}{n}E_{\alpha,\theta}(\psi^{\otimes n})$).

The second one is \cite{bugar2022interpolating}
\begin{equation}\label{eq:Ealphathetacharacterization}
E^{\alpha,\theta}(\psi)=-\lim_{n\to\infty}\frac{1}{n}\sandwiched[\alpha]{\vectorstate{\psi}^{\otimes n}}{A_{\mathcal{H},n}},
\end{equation}
where
\begin{equation}
A_{\mathcal{H},n}=\sum_{\lambda_1,\dots,\lambda_k\vdash n}2^{n\theta(1)\entropy(\lambda_1/n)+\dots+n\theta(k)\entropy(\lambda_k/n)}P^{\mathcal{H}_1}_{\lambda_1}\otimes\dots\otimes P^{\mathcal{H}_k}_{\lambda_k}.
\end{equation}

\section{General upper bound}\label{sec:rhogeneral}

In this section we derive an upper bound on the functionals $E^{\alpha,\theta}$. We make use of the characterisation in \eqref{eq:Ealphathetacharacterization}.
Due to the monotonicity of the minimal R\'enyi divergence, for $\alpha\ge\frac{1}{2}$ and any channel $T$ on $\mathcal{H}$ the inequality
\begin{equation}
\begin{split}
E^{\alpha,\theta}(\psi)
 & \le -\liminf_{n\to\infty}\frac{1}{n}\sandwiched[\alpha]{T^{\otimes n}(\vectorstate{\psi}^{\otimes n})}{T^{\otimes n}(A_{\mathcal{H},n})}  \\
 & = -\liminf_{n\to\infty}\frac{1}{n}\sandwiched[\alpha]{T(\vectorstate{\psi})^{\otimes n}}{T^{\otimes n}(A_{\mathcal{H},n})}  \\
\end{split}
\end{equation}
holds. While this reasoning does not work for $\alpha<\frac{1}{2}$, we can take advantage of the fact that the first argument is a rank-one operator and use the following modified expression:
\begin{equation}\label{eq:Egeneralbound}
\begin{split}
E^{\alpha,\theta}(\psi)
 & = -\lim_{n\to\infty}\frac{1}{n} \sandwiched[\alpha]{\vectorstate{\psi}^{\otimes n}}{A_{\mathcal{H},n}}  \\
 & = -\lim_{n\to\infty}\frac{1}{n}\frac{\alpha}{\alpha-1}\log\bra{\psi}^{\otimes n}A^{\frac{1-\alpha}{\alpha}}_{\mathcal{H},n}\ket{\psi}^{\otimes n}  \\
 & = -\lim_{n\to\infty}\frac{1}{n} \frac{\alpha}{1-\alpha} \sandwiched[\frac{1}{2}]{\vectorstate{\psi}^{\otimes n}}{A^{\frac{1-\alpha}{\alpha}}_{\mathcal{H},n}}  \\
 & \le -\liminf_{n\to\infty}\frac{1}{n} \frac{\alpha}{1-\alpha} \sandwiched[\frac{1}{2}]{T(\vectorstate{\psi})^{\otimes n}}{T^{\otimes n}(A^{\frac{1-\alpha}{\alpha}}_{\mathcal{H},n})}
\end{split}
\end{equation}
Our goal is to apply this inequality in the special case when $T$ is a measurement in a product basis and to find a single-letter expression for the liminf (which turns out to be a limit).

We start with a bound on a single measured Schur--Weyl projection (with respect to a tensor power basis). To this end, let $\mathcal{H}$ be a Hilbert space, $\{\ket{x}\}_{x\in\mathcal{X}}$ an orthonormal basis, and
\begin{equation}\label{eq:measurementchannel}
\mathcal{M}(\rho)=\sum_{x\in\mathcal{X}}\ketbra{x}{x}\rho\ketbra{x}{x}
\end{equation}
the corresponding measurement channel.
\begin{lemma}\label{lem:measuredSchurWeyl}
For $\lambda\vdash n$, let $P_\lambda\in\boundeds(\mathcal{H}^{\otimes n})$ denote the orthogonal projection on the corresponding term in the Schur--Weyl decomposition. Then
\begin{align}
\mathcal{M}^{\otimes n}(P_\lambda)
 & \le (n+1)^{d(d+1)/2}\sum_{\substack{Q\in\distributions[n](\mathcal{X})  \\  \lambda\majorizes nQ}}2^{n(\entropy(\lambda/n)-\entropy(Q))}\typeclassprojection{n}{Q}  \label{eq:measuredSchurWeylupperbound}  \\
\intertext{and}
\mathcal{M}^{\otimes n}(P_\lambda)
 & \ge (n+d)^{-(d+2)(d-1)/2}\sum_{\substack{Q\in\distributions[n](\mathcal{X})  \\  \lambda\majorizes nQ}}2^{n(\entropy(\lambda/n)-\entropy(Q))}\typeclassprojection{n}{Q}.  \label{eq:measuredSchurWeyllowerbound}
\end{align}
In particular, using the notation introduced in \cref{sec:preliminaries} we have
\begin{equation}
\mathcal{M}^{\otimes n}(P_\lambda)\polyeq\sum_{\substack{Q\in\distributions[n](\mathcal{X})  \\  \lambda\majorizes nQ}}2^{n(\entropy(\lambda/n)-\entropy(Q))}\typeclassprojection{n}{Q}.
\end{equation}
\end{lemma}
\begin{proof}
Since $P_\lambda$ is invariant under the action of $S_n$ permuting the tensor factors of $\mathcal{H}^{\otimes n}$ and $\mathcal{M}^{\otimes n}$ is equivariant, the result is also $S_n$-invariant in addition to being diagonal, therefore it is a linear combination of the type class projections $\typeclassprojection{n}{Q}$. More specifically, since $\Tr\typeclassprojection{n}{Q}P_\lambda$ is the dimension of the space of weight-$nQ$ vectors in $P_\lambda\mathcal{H}^{\otimes n}\simeq[\lambda]\otimes\mathbb{S}_\lambda(\mathcal{H})$, we have
\begin{equation}
\begin{split}
\mathcal{M}^{\otimes n}(P_\lambda)
 & = \sum_{Q\in\distributions[n](\mathcal{X})}\frac{\Tr\typeclassprojection{n}{Q}P_\lambda}{\lvert\typeclass{n}{Q}\rvert}\typeclassprojection{n}{Q}
 & = \sum_{Q\in\distributions[n](\mathcal{X})}\frac{\dim[\lambda] K_{\lambda,nQ^{\downarrow}}}{\lvert\typeclass{n}{Q}\rvert}\typeclassprojection{n}{Q}.
\end{split}
\end{equation}
Using \cref{eq:unitarydimensionbound,eq:symmetricdimensionbound,eq:typeclassbound,eq:Kostkabound,eq:Kostkavanishing}, we obtain \cref{eq:measuredSchurWeylupperbound,eq:measuredSchurWeyllowerbound}.
\end{proof}

In the following we choose an orthonormal basis in each tensor factor of $\mathcal{H}=\mathcal{H}_1\otimes\dots\otimes\mathcal{H}_k$ (with index sets $\mathcal{X}_1,\dots,\mathcal{X}_k$) and let $\mathcal{M}_1,\dots,\mathcal{M}_k$ denote the corresponding measurement channels constructed as in \eqref{eq:measurementchannel}. We also introduce the notation $\mathcal{M}=\mathcal{M}_1\otimes\dots\otimes\mathcal{M}_k$.
\begin{proposition}\label{prop:measuredAasymptotics}
\begin{equation}\label{eq:measuredAasymptotics}
\mathcal{M}^{\otimes n}(A_{\mathcal{H},n}^{\frac{1-\alpha}{\alpha}})
  \polyeq \sum_{Q\in\distributions[n](\mathcal{X}_1\times\dots\times\mathcal{X}_k)}2^{\frac{1-\alpha}{\alpha}n\sum_{j=1}^k\theta(j)\entropy(Q_j)}\typeclassprojection{n}{Q}.
\end{equation}
\end{proposition}
\begin{proof}

Note that
\begin{equation}
A_{\mathcal{H},n}^{\frac{1-\alpha}{\alpha}}
 = \bigotimes_{j=1}^k\sum_{\lambda\vdash n}2^{\frac{1-\alpha}{\alpha}n\theta(j)\entropy(\lambda)}P^{\mathcal{H}_j}_\lambda
\end{equation}
and therefore
\begin{equation}\label{eq:measuredAproductform}
\begin{split}
\mathcal{M}^{\otimes n}(A_{\mathcal{H},n}^{\frac{1-\alpha}{\alpha}})
 & = (\mathcal{M}_1^{\otimes n}\otimes\dots\otimes\mathcal{M}_k^{\otimes n})(A_{\mathcal{H},n}^{\frac{1-\alpha}{\alpha}})  \\
 & = \bigotimes_{j=1}^k\sum_{\lambda\vdash n}2^{\frac{1-\alpha}{\alpha}n\theta(j)\entropy(\lambda)}\mathcal{M}_j^{\otimes n}(P^{\mathcal{H}_j}_\lambda).
\end{split}
\end{equation}

We estimate the factors using \cref{lem:measuredSchurWeyl} as
\begin{equation}
\begin{split}
\sum_{\lambda\vdash n}2^{\frac{1-\alpha}{\alpha}n\theta(j)\entropy(\lambda)}\mathcal{M}_j^{\otimes n}(P^{\mathcal{H}_j}_\lambda)
 & \polyeq \sum_{\lambda\vdash n}2^{\frac{1-\alpha}{\alpha}n\theta(j)\entropy(\lambda)}\sum_{\substack{Q\in\distributions[n](\mathcal{X}_j)  \\  \lambda\majorizes nQ}}2^{n(\entropy(\lambda/n)-\entropy(Q))}\typeclassprojection{n}{Q}  \\
 & = \sum_{Q\in\distributions[n](\mathcal{X}_j)}\sum_{\substack{\lambda\vdash n  \\  \lambda\majorizes nQ}}2^{\frac{1-\alpha}{\alpha}n\theta(j)\entropy(\lambda)}2^{n(\entropy(\lambda/n)-\entropy(Q))}\typeclassprojection{n}{Q}.
\end{split}
\end{equation}
Since the Shannon entropy is Schur concave, the largest term in the coefficient of $\typeclassprojection{n}{Q}$ is obtained by setting $\lambda=nQ^\downarrow$. The number of terms is bounded by $(n+1)^{d_1+\dots+d_k}$, therefore
\begin{equation}
\sum_{\lambda\vdash n}2^{\frac{1-\alpha}{\alpha}n\theta(j)\entropy(\lambda)}\mathcal{M}_j^{\otimes n}(P^{\mathcal{H}_j}_\lambda)\polyeq\sum_{Q\in\distributions[n](\mathcal{X}_j)}2^{\frac{1-\alpha}{\alpha}n\theta(j)\entropy(Q)}\typeclassprojection{n}{Q},
\end{equation}

By \eqref{eq:measuredAproductform} we have
\begin{equation}
\begin{split}
\mathcal{M}^{\otimes n}(A_{\mathcal{H},n}^{\frac{1-\alpha}{\alpha}})
 & \polyeq \sum_{Q_1,\dots,Q_k}2^{\frac{1-\alpha}{\alpha}n\sum_{j=1}^k\theta(j)\entropy(Q)}\typeclassprojection{n}{Q_1}\otimes\dots\otimes\typeclassprojection{n}{Q_k}  \\
 & = \sum_{Q\in\distributions[n](\mathcal{X}_1\times\dots\times\mathcal{X}_k)}2^{\frac{1-\alpha}{\alpha}n\sum_{j=1}^k\theta(j)\entropy(Q_j)}\typeclassprojection{n}{Q}
\end{split}
\end{equation}
as claimed.
\end{proof}

\begin{proposition}\label{prop:measuredEbound}
\begin{equation}\label{eq:measuredEbound}
\begin{split}
E^{\alpha,\theta}(\psi)
 & \le -\lim_{n\to\infty}\frac{1}{n} \frac{\alpha}{1-\alpha} \sandwiched[\frac{1}{2}]{\mathcal{M}(\vectorstate{\psi})^{\otimes n}}{\mathcal{M}^{\otimes n}(A^{\frac{1-\alpha}{\alpha}}_{\mathcal{H},n})}  \\
 & = \max_{Q\in\distributions(\supp\mathcal{M}(\vectorstate{\psi})}\left[\sum_{j=1}^k\theta(j)\entropy(Q_j)-\frac{\alpha}{1-\alpha}\relativeentropy{Q}{\mathcal{M}(\vectorstate{\psi})}\right].
\end{split}
\end{equation}
\end{proposition}
\begin{proof}
We only need to establish the existence of the limit and the equality, then the inequality follows from \eqref{eq:Egeneralbound}. Since the image of $\mathcal{M}^{\otimes n}$ consists of commuting operators, the divergence reduces to the classical R\'enyi entropy
\begin{equation}
\sandwiched[\frac{1}{2}]{\mathcal{M}(\vectorstate{\psi})^{\otimes n}}{\mathcal{M}^{\otimes n}(A^{\frac{1-\alpha}{\alpha}}_{\mathcal{H},n})}=-2\log\Tr\left((\mathcal{M}(\vectorstate{\psi})^{\otimes n})^{1/2}(\mathcal{M}^{\otimes n}(A^{\frac{1-\alpha}{\alpha}}_{\mathcal{H},n}))^{1/2}\right).
\end{equation}
 We decompose $\mathcal{M}(\vectorstate{\psi})^{\otimes n}$ into type classes and use \cite[Lemma 2.6]{csiszar2011information} to get
\begin{equation}
\begin{split}
\mathcal{M}(\vectorstate{\psi})^{\otimes n}
 & = \sum_{Q\in\distributions[n](\mathcal{X}_1\times\dots\times\mathcal{X}_k)}\frac{\Tr\mathcal{M}(\vectorstate{\psi})^{\otimes n}\typeclassprojection{n}{Q}}{\Tr\typeclassprojection{n}{Q}}\typeclassprojection{n}{Q}  \\
 & \polyeq \sum_{Q\in\distributions[n](\mathcal{X}_1\times\dots\times\mathcal{X}_k)}\frac{2^{-n\relativeentropy{Q}{\mathcal{M}(\vectorstate{\psi})}}}{\Tr\typeclassprojection{n}{Q}}\typeclassprojection{n}{Q}.
\end{split}
\end{equation}
Combining this estimate with \cref{prop:measuredAasymptotics},
\begin{equation}
\begin{split}
\Tr\left((\mathcal{M}(\vectorstate{\psi})^{\otimes n})^{1/2}(\mathcal{M}^{\otimes n}(A^{\frac{1-\alpha}{\alpha}}_{\mathcal{H},n}))^{1/2}\right)
 & \polyeq \sum_{Q\in\distributions[n](\mathcal{X}_1\times\dots\times\mathcal{X}_k)}2^{\frac{1-\alpha}{2\alpha}n\sum_{j=1}^k\theta(j)\entropy(Q_j)}\frac{2^{-\frac{1}{2}n\relativeentropy{Q}{\mathcal{M}(\vectorstate{\psi})}}}{\Tr\typeclassprojection{n}{Q}}\Tr\typeclassprojection{n}{Q}  \\
 & = \sum_{Q\in\distributions[n](\mathcal{X}_1\times\dots\times\mathcal{X}_k)}2^{\frac{1-\alpha}{2\alpha}n\sum_{j=1}^k\theta(j)\entropy(Q_j)}2^{-\frac{1}{2}n\relativeentropy{Q}{\mathcal{M}(\vectorstate{\psi})}}  \\
 & \polyeq \max_{Q\in\distributions(\supp\mathcal{M}(\vectorstate{\psi})}2^{\frac{1-\alpha}{2\alpha}n\sum_{j=1}^k\theta(j)\entropy(Q_j)-\frac{1}{2}n\relativeentropy{Q}{\mathcal{M}(\vectorstate{\psi})}},
\end{split}
\end{equation}
in the last step using that there are polynomially many type classes, that any distribution can be approximated by $n$-types as $n$ grows and that the exponent is continuous in $Q$. The claim follows by taking the logarithm, multiplying by $\frac{1}{n}\frac{2\alpha}{1-\alpha}$ and letting $n\to\infty$.
\end{proof}

Observe that the right hand side of \eqref{eq:measuredEbound} is a function of the classical distribution $\mathcal{M}(\vectorstate{\psi})$. As our next goal is to study this quantity, it will be convenient to define a short notation:
\begin{definition}\label{def:Halphatheta}
For a measure $P$ on a product of $k$ finite sets, $\alpha\in[0,1)$ and $\theta\in\distributions([k])$ we set
\begin{equation}
\entropy_{\alpha,\theta}(P)=\max_{Q\in\distributions(\supp P)}\left[\sum_{j=1}^k\theta(j)\entropy(Q_j)-\frac{\alpha}{1-\alpha}\relativeentropy{Q}{P}\right].
\end{equation}
\end{definition}
As $\alpha\to 0$, the relative entropy term vanishes, and we obtain the parameter
\begin{equation}
\lim_{\alpha\to 0}\entropy_{\alpha,\theta}(P)=\max_{Q\in\distributions(\supp P)}\sum_{j=1}^k\theta(j)\entropy(Q_j)
\end{equation}
first considered by Strassen \cite[eq. (2.2)]{strassen1991degeneration}, which depends only on the support. When $P$ is a probability distribution and $\alpha\to 1$, the negative relative entropy term dominates unless $Q=P$, therefore
\begin{equation}
\lim_{\alpha\to 1}\entropy_{\alpha,\theta}(P)=\sum_{j=1}^k\theta(j)\entropy(P_j).
\end{equation}

If $\theta(j)=1$ for some $j\in[k]$, then $\entropy_{\alpha,\theta}(P)=\entropy_\alpha(P_j)=\frac{1}{1-\alpha}\log\norm[1]{P^\alpha}$. To see this, we apply \eqref{eq:variationalRenyi} to the marginal distribution to obtain
\begin{equation}\label{eq:entrAlphaIdentity}
\entropy_\alpha(P_j)=\max_{\hat{Q}\in\supp P_j}\left[\entropy(\hat{Q})-\frac{\alpha}{1-\alpha}\relativeentropy{\hat{Q}}{P_j}\right]
\end{equation}
and use that for any $\hat{Q}\in\distributions(\supp P_j)$ we may form the probability distribution $Q(i_1,\dots,i_k)=P(i_1,\dots,i_k)\frac{\hat{Q}(i_j)}{P_j(i_j)}$ that satisfies
\begin{equation}
\relativeentropy{Q}{P}=\relativeentropy{\hat{Q}}{P_j}
\end{equation}
and has $j$th marginal equal to $\hat{Q}$.

Since the relative entropy is jointly convex and $\entropy(Q_j)$ is a concave function of $Q$, $\entropy_{\alpha,\theta}(P)$ is the maximum of a jointly concave function, therefore it is also concave. On the other hand, for fixed $\alpha$ and $P$, $\theta\mapsto\entropy_{\alpha,\theta}(P)$ is the pointwise maximum of a set of affine functions, therefore it is convex.

Using these properties we arrive at the following simple upper and lower bounds on $\entropy_{\alpha,\theta}(P)$:
\begin{proposition}
Let $\alpha\in[0,1]$, $\theta\in\distributions([k])$ and $P$ a measure on $\mathcal{X}_1\times\dots\times\mathcal{X}_k$. Then
\begin{align}
\entropy_{\alpha,\theta}(P) & \ge \sum_{j=1}^k\theta(j)\entropy(P_j/\norm[1]{P})+\frac{\alpha}{1-\alpha}\log\norm[1]{P}  \\
\intertext{and}
\entropy_{\alpha,\theta}(P)
 & \le \sum_{j=1}^k\theta(j)\entropy_\alpha(P_j/\norm[1]{P})+\frac{\alpha}{1-\alpha}\log\norm[1]{P}  \\
 & \le \sum_{j=1}^k\theta(j)\log\lvert\supp P_j\rvert+\frac{\alpha}{1-\alpha}\log\norm[1]{P}  \label{eq:Halphathetacrudeupperbound}
\end{align}
\begin{equation}
\end{equation}
\end{proposition}
\begin{proof}
To prove the lower bound, we evaluate the objective function at the feasible point $Q=P/\norm[1]{P}$:
\begin{equation}
\sum_{j=1}^k\theta(j)\entropy(P_j/\norm[1]{P})-\frac{\alpha}{1-\alpha}\relativeentropy{P/\norm[1]{P}}{P}=\sum_{j=1}^k\theta(j)\entropy(P_j/\norm[1]{P})+\frac{\alpha}{1-\alpha}\log\norm[1]{P}.
\end{equation}

The first upper bound follows from convexity in the parameter $\theta$, writing $\theta$ as a convex combination of the $k$ Dirac measures and using that at these points the functional reduces to the marginal R\'enyi entropies, which satisfy the scaling law
\begin{equation}
\entropy_\alpha(P_j)=\entropy_\alpha(P_j/\norm[1]{P})+\frac{\alpha}{1-\alpha}\log\norm[1]{P}.
\end{equation}
The second upper bound is true because $\entropy_\alpha(\hat{P})\le\log\lvert\supp\hat{P}\rvert$ for any probability distribution $\hat{P}$.
\end{proof}

We will make use of the following refined upper bound that depends on a block-decomposition of the measure, and is given in terms of the sizes and measures of the blocks. If we choose a single block, the bound reduces to \eqref{eq:Halphathetacrudeupperbound}.
\begin{proposition}\label{prop:Halphathetarefinedupperbound}
Let $P$ be a measure on $\mathcal{X}_1\times\dots\times\mathcal{X}_k$. For each $j\in[k]$ let $(I_{j,i})_{i=1}^{m_j}$ be a partition of the set $\mathcal{X}_j$. Then
\begin{equation}
\entropy_{\alpha,\theta}(P)\le\sum_{j=1}^k\theta(j)\log m_j+\frac{\alpha}{1-\alpha}\sum_{j=1}^k\log m_j+\max_{i_1,\dots,i_k}\left[\sum_{j=1}^k\theta(j)\log\lvert I_{j,{i_j}}\rvert+\frac{\alpha}{1-\alpha}\log P(I_{1,i_1}\times\dots\times I_{k,i_k})\right].
\end{equation}
\end{proposition}
\begin{proof}
Let $U$ denote the average of the measures $(\sigma_1\times\dots\times\sigma_k)(P)$ as $\sigma_j$ ranges over all permutations that setwise fix the subsets $I_{j,i}$. The product permutations do not change the value of $\entropy_{\alpha,\theta}$, therefore by concavity we have $\entropy_{\alpha,\theta}(P)\le\entropy_{\alpha,\theta}(U)$. In the same way,
\begin{equation}
Q\mapsto\sum_{j=1}^k\theta(j)\entropy(Q_j)-\frac{\alpha}{1-\alpha}\relativeentropy{Q}{U}
\end{equation}
is concave and invariant under product permutations fixing the subsets, therefore the maximum is attained at a distribution $Q$ that is uniform within each block $I_{1,i_1}\times\dots\times I_{k,i_k}$. Let $\hat{P}(i_1,\dots,i_k)=P(I_{1,i_1}\times\dots\times I_{k,i_k})$ and $\hat{Q}(i_1,\dots,i_k)=Q(I_{1,i_1}\times\dots\times I_{k,i_k})$ be the measures induced on the set of blocks. Then
\begin{equation}
\begin{split}
\sum_{j=1}^k\theta(j)\entropy(Q_j)-\frac{\alpha}{1-\alpha}\relativeentropy{Q}{U}
 & = \sum_{j=1}^k\theta(j)\left(\entropy(\hat{Q}_j)+\sum_{i=1}^{m_j}\hat{Q}_j(i)\log\lvert I_{j,i}\rvert\right)-\frac{\alpha}{1-\alpha}\relativeentropy{\hat{Q}}{\hat{P}}  \\
 & \le \sum_{j=1}^k\theta(j)\log m_j+\sum_{j=1}^k\theta(j)\sum_{i=1}^{m_j}\hat{Q}_j(i)\log\lvert I_{j,i}\rvert-\frac{\alpha}{1-\alpha}\relativeentropy{\hat{Q}}{\hat{P}}  \\
 & = \sum_{j=1}^k\theta(j)\log m_j+\sum_{i_1,\dots,i_k}\hat{Q}(i_1,\dots,i_k)\sum_{j=1}^k\theta(j)\log\lvert I_{j,i_j}\rvert-\frac{\alpha}{1-\alpha}\relativeentropy{\hat{Q}}{\hat{P}}.
\end{split}
\end{equation}
The maximum of the upper bound over $\hat{Q}$ is
\begin{equation}
\sum_{j=1}^k\theta(j)\log m_j+\frac{\alpha}{1-\alpha}\log\sum_{i_1,\dots,i_k}2^{\frac{1-\alpha}{\alpha}\sum_{j=1}^k\theta(j)\log\lvert I_{j,{i_j}}\rvert}\hat{P}(i_1,\dots,i_k),
\end{equation}
which is at most
\begin{equation}
\sum_{j=1}^k\theta(j)\log m_j+\frac{\alpha}{1-\alpha}\sum_{j=1}^k\log m_j+\max_{i_1,\dots,i_k}\left[\sum_{j=1}^k\theta(j)\log\lvert I_{j,{i_j}}\rvert+\frac{\alpha}{1-\alpha}\log\hat{P}(i_1,\dots,i_k)\right].
\end{equation}
\end{proof}

\begin{proposition}\label{prop:Halphathetaproduct}
For every $\alpha\in[0,1]$ and $\theta\in\distributions([k])$ the parameter $\entropy_{\alpha,\theta}$ is additive in the sense that if $P^{(i)}$ is a measure on $\mathcal{X}^{(i)}_1\times\dots\times\mathcal{X}^{(i)}_k)$ (for $i=1,2$) and their product is regarded as a measure on $(\mathcal{X}^{(1)}_1\times\mathcal{X}^{(2)}_1)\times\dots\times(\mathcal{X}^{(1)}_k\times\mathcal{X}^{(2)}_k)$ then
\begin{equation}
\entropy_{\alpha,\theta}(P^{(1)}\otimes P^{(2)})=\entropy_{\alpha,\theta}(P^{(1)})+\entropy_{\alpha,\theta}(P^{(2)}).
\end{equation}
\end{proposition}
\begin{proof}
Let $Q\in\distributions((\mathcal{X}^{(1)}_1\times\mathcal{X}^{(2)}_1)\times\dots\times(\mathcal{X}^{(1)}_k\times\mathcal{X}^{(2)}_k))$ and let $Q^{(1)}$ and $Q^{(2)}$ denote its marginals on $\mathcal{X}^{(1)}_1\times\dots\times\mathcal{X}^{(1)}_k$ and $\mathcal{X}^{(2)}_1\times\dots\times\mathcal{X}^{(2)}_k$, respectively. By the subadditivity of the Shannon entropy and \eqref{eq:KLproductidentity} we have the inequality
\begin{multline}
\sum_{j=1}^k\theta(j)\entropy(Q_j)-\frac{\alpha}{1-\alpha}\relativeentropy{Q}{P^{(1)}\otimes P^{(2)}}  \\
\begin{aligned}
 & = \sum_{j=1}^k\theta(j)\entropy(Q_j)-\frac{\alpha}{1-\alpha}\left(\relativeentropy{Q^{(1)}}{P^{(1)}}+\relativeentropy{Q^{(2)}}{P^{(2)}}+\mutualinformation((1):(2))_Q\right)  \\
 & \le \sum_{j=1}^k\theta(j)\left(\entropy(Q^{(1)}_j)+\entropy(Q^{(2)}_j)\right)-\frac{\alpha}{1-\alpha}\left(\relativeentropy{Q^{(1)}}{P^{(1)}}+\relativeentropy{Q^{(2)}}{P^{(2)}}\right)  \\
 & \le \entropy_{\alpha,\theta}(P^{(1)})+\entropy_{\alpha,\theta}(P^{(2)}),
\end{aligned}
\end{multline}
with equality iff $Q=Q^{(1)}\otimes Q^{(2)}$ and $Q^{(1)}$ and $Q^{(2)}$ are the maximizing distributions according to \cref{def:Halphatheta}. The maximum of the left hand side over $Q$ (which by definition is $\entropy_{\alpha,\theta}(P^{(1)}\otimes P^{(2)})$) is therefore equal to $\entropy_{\alpha,\theta}(P^{(1)})+\entropy_{\alpha,\theta}(P^{(2)})$.
\end{proof}

\begin{proposition}\label{prop:expHalphathetasum}
Under the conditions of \cref{prop:Halphathetaproduct} the equality
\begin{equation}
2^{(1-\alpha)\entropy_{\alpha,\theta}(P^{(1)}\oplus P^{(2)})}=2^{(1-\alpha)\entropy_{\alpha,\theta}(P^{(1)})}+2^{(1-\alpha)\entropy_{\alpha,\theta}(P^{(2)})}
\end{equation}
holds.
\end{proposition}
\begin{proof}
Any $Q\in\distributions(\supp P^{(1)}\sqcup\supp P^{(2)})$ can be written as $Q=qQ^{(1)}\oplus(1-q)Q^{(2)}$ with $q\in[0,1]$ and $Q^{(i)}\in\distributions(\supp P^{(i)})$. Therefore
\begin{equation}
\begin{split}
\entropy_{\alpha,\theta}(P^{(1)}\oplus P^{(2)})
 & = \max_{q\in[0,1]}\max_{Q^{(1)},Q^{(2)}}\Bigg[\sum_{j=1}^k\theta(j)\entropy(qQ^{(1)}_j\oplus (1-q)Q^{(2)}_j)  \\
 & \qquad -\frac{\alpha}{1-\alpha}\relativeentropy{qQ^{(1)}\oplus(1-q)Q^{(2)}}{P^{(1)}\oplus P^{(2)}}\Bigg]  \\
 & = \max_{q\in[0,1]}\max_{Q^{(1)},Q^{(2)}}\Bigg[\sum_{j=1}^k\theta(j)\left(q\entropy(Q^{(1)}_j)+(1-q)\entropy(Q^{(2)}_j)+h(q)\right)  \\
 & \qquad -\frac{\alpha}{1-\alpha}\left(q\relativeentropy{Q^{(1)}}{P^{(1)}}+(1-q)\relativeentropy{Q^{(2)}}{P^{(2)}}-h(q)\right)\Bigg]  \\
 & = \max_{q\in[0,1]}\left[q\entropy_{\alpha,\theta}(P^{(1)})+(1-q)\entropy_{\alpha,\theta}(P^{(2)})+\frac{1}{1-\alpha}h(q)\right]  \\
 & = \frac{1}{1-\alpha}\log\left(2^{(1-\alpha)\entropy_{\alpha,\theta}(P^{(1)})}+2^{(1-\alpha)\entropy_{\alpha,\theta}(P^{(2)})}\right),
\end{split}
\end{equation}
in the last step using $\max_{q\in[0,1]}(qx+(1-q)y+h(q)=\log(2^x+2^y)$ (see \cite[eq. (2.13)]{strassen1991degeneration}).
\end{proof}

\begin{lemma}\label{lem:differentiateMaximum}
Let $X\subseteq\reals^n$ be compact and convex, $f,g:X\to\reals$ continuous and concave, $g$ strictly concave. The set $M:=f^{-1}(\max_{x\in X}f(x))\subseteq X$ is compact and convex, therefore there is a unique maximizer $\bar{x}\in M$ of $g$. For each $t\in\positivereals$, let $x^*(t)\in X$ be the unique maximizer of $f+tg$ and $v(t)=f(x^*(t))+tg(x^*(t))$ the maximum as a function of $t$. Then $\lim_{t\to 0+}x^*(t)=\bar{x}$, $v$ extends continuously to $\nonnegativereals$ by setting $v(0)=f(\bar{x})$, and $\left.\frac{\ed}{\ed t}v(t)\right|_{t=0}=g(\bar{x})$.
\end{lemma}
\begin{proof}
Let $(t_n)_{n\in\naturals}$ be a sequence in $\positivereals$ converging to $0$. By compactness and passing to a subsequence if necessary, we may assume that $x^*(t_n)$ has a limit $x\in X$. Since $\bar{x}$ maximizes $f$ and $x^*(t_n)$ maximizes $f+t_ng$, we have the chain of inequalitites
\begin{equation}
f(\bar{x})+t_ng(\bar{x})\le \underbrace{f(x^*(t_n))+t_ng(x^*(t_n))}_{v(t_n)}\le f(\bar{x})+t_ng(x^*(t_n)).
\end{equation}
We take the limit $n\to\infty$ and use the continuity of $f$ and $g$ to obtain $f(\bar{x})\le f(x)\le f(\bar{x})$, therefore $x\in M$. From the inequality we also see that $g(\bar{x})\le g(x^*(t_n))$, and taking the limit here gives $g(\bar{x})\le g(x)$. Since $\bar{x}\in M$ is the unique maximizer, we have $x=\bar{x}$. It also follows that $\lim_{t\to 0+}v(t)=f(\bar{x})$.

For the last statement we first note that $v(t)\ge f(\bar{x})+tg(\bar{x})$ since $v(t)$ is a maximum and $\bar{x}$ is a feasible point. On the other hand, $t\mapsto v(t)$ is convex since it is the pointwise maximum of a family of affine functions, therefore for any $0<t<T$ we have
\begin{equation}
\begin{split}
v(t)
 & \le v(0)+\frac{t}{T}(v(T)-v(0))  \\
 & = f(\bar{x})+\frac{t}{T}(f(x^*(T))+Tg(x^*(T))-f(\bar{x}))  \\
 & \le f(\bar{x})+tg(x^*(T)).
\end{split}
\end{equation}
The claim follows since $g(x^*(T))\to g(\bar{x})$ as $T\to 0$.
\end{proof}

\begin{corollary}\label{cor:difHalphathetaByalphaAtZero}
For $\alpha\in (0,1)$ let $Q_\alpha$ be the probability distribution which is optimal in the definition of $H_{\alpha,\theta}(P)$. By the concavity of $Q\mapsto \sum_j\theta(j)H(Q_j)$ and the strict convexity of $\relativeentropy{\cdot}{P}$ it follows that $Q_\alpha$ is unique. On the other hand, for $\alpha=0$, the functional $H_{0,\theta}(P)$ does not need to have a unique maximizing distribution. Let $M=\{ Q\in\mathcal{P}(\supp P) | \sum_j\theta(j)H(Q_j) = H_{0,\theta}(P) \}$ be the set of maximizing distributions for $\alpha=0$ and $Q_0=\argmin_{Q\in M} \relativeentropy{Q}{P}$. 
Applying \cref{lem:differentiateMaximum} to $H_{\alpha,\theta}$ we get
\begin{equation}
    \left.\frac{\partial}{\partial\alpha} H_{\alpha,\theta}(P)\right|_{\alpha=0}= \left.-\frac{1}{(1-\alpha)^2}\relativeentropy{Q_\alpha}{P}\right|_{\alpha=0}=
    -\relativeentropy{Q_0}{P}.
\end{equation}
\end{corollary}

Recall that the bound \eqref{eq:measuredEbound} depends on the chosen local von Neumann measurements $\mathcal{M}_1,\dots,\mathcal{M}_k$. To obtain the strongest possible upper bound, one should choose the local measurements that minimize the right hand side. We now shift our focus to the resulting optimized bound which we denote by $\rho^{\alpha,\theta}$ because of its similarity to the logarithmic upper support functional $\rho^\theta$ \cite[eq. (2.4)]{strassen1991degeneration}.
\begin{definition}\label{def:alphasupport}
Let $\alpha\in(0,1)$ and $\theta\in\distributions([k])$. We define the functional $\rho^{\alpha,\theta}$ by
\begin{equation}
\rho^{\alpha,\theta}(\psi)=\min_{\mathcal{M}_1,\dots,\mathcal{M}_k}\entropy_{\alpha,\theta}((\mathcal{M}_1\otimes\dots\otimes\mathcal{M}_k)(\vectorstate{\psi})),
\end{equation}
where the minimum is over local measurement channels $\mathcal{M}_j:\boundeds(\mathcal{H}_j)\to\boundeds(\mathcal{H}_j)$, and the exponentiated form $\zeta^{\alpha,\theta}(\psi)=2^{(1-\alpha)\rho^{\alpha,\theta}(\psi)}$.
\end{definition}
\begin{corollary}\label{cor:rhoalphathetaproperties}\leavevmode
\begin{enumerate}
\item\label{it:ElessRho} $E^{\alpha,\theta}(\psi)\le\rho^{\alpha,\theta}(\psi)$.
\item\label{it:rhosubmulti} $\rho^{\alpha,\theta}(\psi\otimes\varphi)\le\rho^{\alpha,\theta}(\psi)+\rho^{\alpha,\theta}(\varphi)$
\item\label{it:zetasubadd} $\zeta^{\alpha,\theta}(\psi\oplus\varphi)\le \zeta^{\alpha,\theta}(\psi)+\zeta^{\alpha,\theta}(\varphi)$
\end{enumerate}
\end{corollary}
\begin{proof}

\ref{it:ElessRho}: This is a direct consequence of \cref{prop:measuredEbound}. That holds for any measurement channel $\mathcal{M}$, so the only remaining step is to take its minimum in $\mathcal{M}$.

\ref{it:rhosubmulti}: This comes from a straightforward calculation. First we restrict the optimization to such channels that are in tensor product form on the given tensor product space:
\begin{equation}
\begin{split}
    \rho^{\alpha,\theta}(\psi\otimes\varphi) &\le 
    \min_{\mathcal{M}_1,\dots,\mathcal{M}_k}\min_{\mathcal{N}_1,\dots,\mathcal{N}_k}\entropy_{\alpha,\theta}((\mathcal{M}_1\otimes\dots\otimes\mathcal{M}_k)\otimes (\mathcal{N}_1\otimes\dots\otimes\mathcal{N}_k)
    \vectorstate{\psi\otimes\varphi})
    \\
    &=\min_{\mathcal{M}_1,\dots,\mathcal{M}_k}\entropy_{\alpha,\theta}(\mathcal{M}_1\otimes\dots\otimes\mathcal{M}_k\vectorstate{\psi})+
    \min_{\mathcal{N}_1,\dots,\mathcal{N}_k}\entropy_{\alpha,\theta}( \mathcal{N}_1\otimes\dots\otimes\mathcal{N}_k\vectorstate{\varphi})\\
    &=\rho^{\alpha,\theta}(\psi)+\rho^{\alpha,\theta}(\varphi).
    \end{split}
\end{equation}
In the first equality we used \cref{prop:Halphathetaproduct}.

\ref{it:zetasubadd}: Similarly we restrict the optimization to such channels that are in the form $\mathcal{M}\oplus\mathcal{N}$ where $\mathcal{M}$ acts on $\vectorstate{\psi}$ and $\mathcal{N}$ acts on $\vectorstate{\varphi}$.
\begin{equation}
    \begin{split}
        \zeta^{\alpha,\theta}(\psi\oplus\varphi)&\le
        \min_{\mathcal{M}_1,\dots,\mathcal{M}_k}\min_{\mathcal{N}_1,\dots,\mathcal{N}_k}
        2^{(1-\alpha)
        \entropy_{\alpha,\theta}((\mathcal{M}_1\otimes\dots\otimes\mathcal{M}_k)\oplus (\mathcal{N}_1\otimes\dots\otimes\mathcal{N}_k)
    \vectorstate{\psi\oplus\varphi})}\\
    &=\min_{\mathcal{M}_1,\dots,\mathcal{M}_k}2^{(1-\alpha)
        \entropy_{\alpha,\theta}(\mathcal{M}_1\otimes\dots\otimes\mathcal{M}_k \vectorstate{\psi})}+
      \min_{\mathcal{N}_1,\dots,\mathcal{N}_k}2^{(1-\alpha)
        \entropy_{\alpha,\theta}(\mathcal{N}_1\otimes\dots\otimes\mathcal{N}_k\vectorstate{\varphi})}\\
        &=\zeta^{\alpha,\theta}(\psi)+\zeta^{\alpha,\theta}(\varphi)
    \end{split}
\end{equation}
In the first equality we use \cref{prop:expHalphathetasum}.

\end{proof}

Since $E^{\alpha,\theta}$ is additive and $\rho^{\alpha,\theta}$ is subadditive under the tensor product, $\frac{1}{n}\rho^{\alpha,\theta}(\psi^{\otimes n})$ gives a sequence of upper bounds on $E^{\alpha,\theta}(\psi)$ that has a limit equal to its infimum. Next we will show that this limit is in fact equal to $E^{\alpha,\theta}(\psi)$.

\begin{theorem}\label{thm:rhoregE}
For all $\psi\in\mathcal{H}_1\otimes\dots\otimes\mathcal{H}_k$, $\alpha\in(0,1)$ and $\theta\in\distributions([k])$ the equality
\begin{equation}
\lim_{n\to\infty}\frac{1}{n}\rho^{\alpha,\theta}(\psi^{\otimes n})=E^{\alpha,\theta}(\psi)
\end{equation}
holds.
\end{theorem}
\begin{proof}
Since $\rho^{\alpha,\theta}$ is subadditive, the limit exists and is not less than the right hand side. We need to show that $\lim_{n\to\infty}\frac{1}{n}\rho^{\alpha,\theta}(\psi^{\otimes n})\le E^{\alpha,\theta}(\psi)$.

Let $n\in\naturals$ and choose local bases that are compatible with the Schur--Weyl decomposition, i.e. that are unions of bases for the subspaces $P^{\mathcal{H}_j}_\lambda\mathcal{H}_j^{\otimes n}$. Let $\mathcal{M}$ be the corresponding measurement channel. The local Schur--Weyl decompositions determine a block decomposition of $\mathcal{M}(\vectorstate{\psi^{\otimes n}})$. Applying \cref{prop:Halphathetarefinedupperbound} with these blocks, we obtain
\begin{equation}\label{eq:rhoboundSchurWeylblock}
\begin{split}
\rho^{\alpha,\theta}(\psi^{\otimes n})
 & \le \entropy_{\alpha,\theta}(\mathcal{M}(\vectorstate{\psi^{\otimes n}}))  \\
 & \le \sum_{j=1}^k\left(\theta(j)+\frac{\alpha}{1-\alpha}\right)\log(n+1)^{d_j}  \\
 & \quad+\max_{\lambda_1,\dots,\lambda_k\vdash n}\left[\sum_{j=1}^k\theta(j)\log\Tr P^{\mathcal{H}_j}_{\lambda_j}+\frac{\alpha}{1-\alpha}\log\norm{(P^{\mathcal{H}_1}_{\lambda_1}\otimes\dots\otimes P^{\mathcal{H}_k}_{\lambda_k})\psi^{\otimes n}}^2\right].
\end{split}
\end{equation}

On the other hand, the product of local contractions $P^{\mathcal{H}_1}_{\lambda_1}\otimes\dots\otimes P^{\mathcal{H}_k}_{\lambda_k}$ provides a lower bound on $E_{\alpha,\theta}(\psi^{\otimes n})$ defined in \eqref{eq:lowerLogLOCC}, since the latter is defined as a supremum over such maps. We estimate the marginal entropies of the resulting state using the same argument as in \cite[Proposition 4.10]{vrana2023family}, which we reproduce here in a slightly modified form.

The group $S_n$ acts on $\mathcal{H}^{\otimes n}$ as well as on $\mathcal{H}_j^{\otimes n}$ by permuting the factors. Since $\psi^{\otimes n}$ is invariant and $P^{\mathcal{H}_j}_{\lambda_j}$ are equivariant (since these are the isotypic projections for representation), the vector $(P^{\mathcal{H}_1}_{\lambda_1}\otimes\dots\otimes P^{\mathcal{H}_k}_{\lambda_k})\psi^{\otimes n}$ is also invariant and so are its marginals. The marginals are supported in a subspace isomorphic to $[\lambda_j]\otimes\mathbb{S}_{\lambda_j}(\mathcal{H}_j)$. The representation $[\lambda_j]$ is irreducible, therefore the reduction of the state to this factor is maximally mixed. On the other hand, the other factor is small: $\dim\mathbb{S}_{\lambda_j}(\mathcal{H}_j)\le(n+1)^{d_j(d_j-1)/2}$. By the triangle inequality for the von Neumann entropy, the $j$th marginal entropy of the normalization of $(P^{\mathcal{H}_1}_{\lambda_1}\otimes\dots\otimes P^{\mathcal{H}_k}_{\lambda_k})\psi^{\otimes n}$ is at least
\begin{equation}
\begin{split}
\log\dim[\lambda_j]-\log\dim\mathbb{S}_{\lambda_j}(\mathcal{H}_j)
 & = \log\Tr P^{\mathcal{H}_j}_{\lambda_j}-2\log\dim\mathbb{S}_{\lambda_j}(\mathcal{H}_j)  \\
 & \ge \log\Tr P^{\mathcal{H}_j}_{\lambda_j}-d_j(d_j-1)\log(n+1).
\end{split}
\end{equation}
This implies
\begin{equation}
E_{\alpha,\theta}(\psi^{\otimes n})\ge\sum_{j=1}^k\theta(j)\left(\log\Tr P^{\mathcal{H}_j}_{\lambda_j}-d_j(d_j-1)\log(n+1)\right)+\frac{\alpha}{1-\alpha}\log\norm{(P^{\mathcal{H}_1}_{\lambda_1}\otimes\dots\otimes P^{\mathcal{H}_k}_{\lambda_k})\psi^{\otimes n}}^2.
\end{equation}
Combining this inequality with \eqref{eq:rhoboundSchurWeylblock}, we have
\begin{equation}
\rho^{\alpha,\theta}(\psi^{\otimes n})\le E_{\alpha,\theta}(\psi^{\otimes n})+\sum_{j=1}^n\theta(j)d_j^2\log(n+1)+\frac{\alpha}{1-\alpha}\sum_{j=1}^kd_j\log(n+1).
\end{equation}
Finally, we divide both sides by $n$ and let $n\to\infty$:
\begin{equation}
\lim_{n\to\infty}\frac{1}{n}\rho^{\alpha,\theta}(\psi^{\otimes n})\le\lim_{n\to\infty}\frac{1}{n}E_{\alpha,\theta}(\psi^{\otimes n})=E^{\alpha,\theta}(\psi).
\end{equation}
\end{proof}

\section{States with free support}\label{sec:rhofree}

We consider a class of states with a property closely related but not identical to the notion of free tensors, introduced in \cite{franz2002moment}. We will show that for any state vector $\psi$ belonging to this class we have $\rho^{\alpha,\theta}(\psi)=E^{\alpha,\theta}(\psi)$ for all $\alpha\in[0,1]$ and $\theta\in\distributions([k])$. For the analogous property of free tensors in the $\alpha=0$ case we refer the reader to \cite[Section 4.2.]{christandl2023universal}.
\begin{definition}\label{def:free}
    Let $\Psi\subseteq \mathcal{X}_1\times\dots\times\mathcal{X}_k$. We say $\Psi$ is free if any two $k$-tuples $x, y \in \Psi$ where $x\neq y$ are different in at least two positions.
	We will say that a state (vector) $\psi$ has free support if there exist local measurement channels $\mathcal{M}_j$ such that $\supp(\mathcal{M}_1\otimes\dots\otimes\mathcal{M}_k)(\vectorstate{\psi})$ is free (as a subset of $\mathcal{X}_1\times\dots\times\mathcal{X}_k$).
\end{definition}
By contrast, a free \emph{tensor} is defined by the same support condition with respect to any (not necessarily orthonormal) product basis, therefore a state vector with free support is also a free tensor but not the other way around. On the other hand, the SLOCC orbit of a free tensor contains at least one state with free support, since we can apply linear maps that map the local basis to a local \emph{orthonormal} basis.
\begin{example}
Every tensor in $\complexes^2\otimes\complexes^2\otimes\complexes^2$ is free. The standard representatives of all 6 SLOCC classes of three-qubit states \cite{dur2000three} have free support: $\ket{000}$ (separable), $\ket{\EPR_{AB}}=\frac{1}{\sqrt{2}}(\ket{000}+\ket{110})$ and its permutations (biseparable), $\ket{\W}=\frac{1}{\sqrt{3}}(\ket{100}+\ket{010}+\ket{001})$, and $\ket{\GHZ}=\frac{1}{\sqrt{2}}(\ket{000}+\ket{111})$. On the other hand, there exist three-qubit states that do not have free support (see \cref{ex:notfree} below).
\end{example}

\begin{remark}
A free subset $\Phi$ of $[d]\times\dots\times[d]$ (with $k$) factors has at most $d^{k-1}$ elements since the projection to (say) the first $k-1$ factors is injective when restricted to $\Phi$. In this sense, states with a free support are sparse, with a sublinear number of nonzero entries in a suitable product basis. By dimension counting, this also implies that the property of having free support is not generic (this is also true for free tensors, see \cite[Remark 4.17.]{christandl2023universal}).

We note that for a given Hilbert space, the set of normalised state vectors having free support is compact. Indeed, for a given free subset $\Phi\subseteq[d_1]\times\dots\times[d_k]$, a probability distribution $P\in\distributions(\Phi)$, phases $z\in U(1)^\Phi$ and local unitaries $U_1,\dots,U_k$ we can form the vector $(U_1\otimes\dots\otimes U_k)\sum_{(x_1,\dots,x_k)\in\Phi}z_{(x_1,\dots,x_k)}\sqrt{P(x_1,\dots,x_k)}\ket{x_1,\dots,x_k}$. This unit vector depends continuously on $(P,z,U_1,\dots,U_k)$ and $\distributions(\Phi)\times U(1)^\Phi\times U(d_1)\times\dots\times U(d_k)$ is compact, therefore the set of vectors obtained in this way is compact. Every normalised state with free support arises in this way for a suitable $\Phi$, of which there are finitely many.
\end{remark}
We note that the measurement channels satisfying the support condition need not be unique. A well-known example is the \GHZ{} state, which has free support with respect to the computational basis, and also with respect to the local X bases (see \cref{ex:GHZ} below).

If $\psi$ and $\varphi$ both have free support then $\psi\otimes\varphi$ and $\psi\oplus\varphi$ also have free support. This can be seen by taking the tensor product (respectively direct sum) of the local bases that correspond to the measurement channels in the definition.

\begin{proposition}\label{prop:marginalOfFree}
Let a linear operator $A:\mathcal{H}\to\mathcal{H}$ be diagonal in a product basis for which the support of $\psi$ is free (i.e., $(\mathcal{M}_1\otimes\dots\otimes\mathcal{M}_k)(A)=A$), then $A\psi$ also has a free support. In particular, if $A_j:\mathcal{H}_j\to\mathcal{H}_j$ are linear operators that are diagonal in the basis defining $\mathcal{M}_j$ (i.e., $\mathcal{M}_j(A_j)=A_j$), then $(A_1\otimes\dots\otimes A_k)\psi$ has a free support with respect to the same measurements.

Let $\psi$ have a free support with respect to the local measurement channels $\mathcal{M}_1,\dots,\mathcal{M}_k$ and let $\mathcal{M}=\mathcal{M}_1\otimes\dots\otimes\mathcal{M}_k$. Then
\begin{equation}
   \Tr_{23\dots}\vectorstate{\psi}  = \Tr_{23\dots}\mathcal{M}(\vectorstate{\psi})
\end{equation}
and similarly for the other single-site marginals. In particular, every single-site marginal of a free state is diagonal with respect to the same bases as in the definition.
\end{proposition}
\begin{proof}
The maps $\mathcal{M}_j$ are trace-preserving, therefore $\mathcal{M}_1(\Tr_{23\dots}\vectorstate{\psi})=\Tr_{23\dots}\mathcal{M}(\vectorstate{\psi})$. We only need to show that $\mathcal{M}_1$ can be omitted, which is equivalent to the statement that the marginal of $\psi$ is already diagonal. The entries of the marginal are
\begin{equation}
\begin{split}
\bra{x_1}\left(\Tr_{23\dots}\vectorstate{\psi}\right)\ket{x'_1}
 & = \sum_{x_2,\dots,x_k}\braket{x_1,x_2,\dots,x_k}{\psi}\braket{\psi}{\smash{x'_1},x_2,\dots,x_k}  \\
 & = \delta_{x_1 x'_1}\sum_{x_2,\dots,x_k} \left|\braket{x_1,x_2,\dots,x_k}{\psi}\right|^2,
\end{split}
\end{equation}
where the second equality is an immediate consequence of the definition of a free support, namely that in this basis two non-zero entries have to differ at least in two indices.
\end{proof}

\Cref{prop:marginalOfFree} can be used as a simple test to decide whether a particular state with generic marginals has free support. In the following example we use this to show that there exist three-qubit states that do not have free support.
\begin{example}\label{ex:notfree}
Let $\ket{\psi}=\frac{1}{\sqrt{10}}(2\ket{000}-\ket{100}-\ket{010}-\ket{001}-\ket{011}-\ket{101}-\ket{110})$. The marginals of $\vectorstate{\psi}$ are diagonal in the computational basis and have distinct eigenvalues ($\frac{7}{10}$ and $\frac{3}{10}$), therefore this is the unique basis (up to phase) that diagonalizes them. Therefore if $\ket{\psi}$ had free support, then by \cref{prop:marginalOfFree} its support with respect to the computational basis would be free. However, this support contains both $(0,0,0)$ and $(1,0,0)$, therefore it is not free.
\end{example}

\begin{theorem}\label{thm:freeErho}
Let $\alpha\in(0,1]$, $\theta\in\distributions([k])$ and let $\psi$ be a state with free support with respect to the local measurements $\mathcal{M}_1,\dots,\mathcal{M}_k$. Then
\begin{equation}
E^{\alpha,\theta}(\psi)=\rho^{\alpha,\theta}(\psi)=\entropy_{\alpha,\theta}((\mathcal{M}_1\otimes\dots\otimes\mathcal{M}_k)(\vectorstate{\psi})).
\end{equation}
In particular, any tuple of local measurements satisfying the support condition in \cref{def:free} is a minimizer in \cref{def:alphasupport}.
\end{theorem}
\begin{proof}
We have seen in \cref{cor:rhoalphathetaproperties} that $E^{\alpha,\theta}(\psi)\le\rho^{\alpha,\theta}(\psi)$ holds for every state $\psi$.

Let $\mathcal{X}_j$ be the index set for the basis defining the local measurement $\mathcal{M}_j$ and let $\mathcal{M}=\mathcal{M}_1\otimes\dots\otimes\mathcal{M}_k$ as before. For some $n\in\naturals$ let $Q\in\distributions[n](\supp\mathcal{M}(\vectorstate{\psi}))$. Its marginals $Q_j\in\distributions[n](\mathcal{X}_j)$ determine the type class projections $\typeclassprojection{n}{Q_j}$ acting on the local $n$-copy Hilbert space $\mathcal{H}_j^{\otimes n}$. We will use these local projections to bound $E_{\alpha,\theta}(\psi^{\otimes n})$ from below.

Let $\Pi=\typeclassprojection{n}{Q_1}\otimes\dots\otimes\typeclassprojection{n}{Q_k}$. From \cref{prop:marginalOfFree} and the fact that the $\typeclassprojection{n}{Q_j}$ are diagonal in the tensor power basis (i.e., $\mathcal{M}_j^{\otimes n}(\typeclassprojection{n}{Q_j})=\typeclassprojection{n}{Q_j}$) it follows that the marginals of $\Pi\psi^{\otimes n}$ are also diagonal. Furthermore, by the transitivity of $S_n$ on the type classes, the $j$-th marginal is a uniform distribution on the image of $\Pi_j$. Then we have $\entropy(\Tr_j \Pi\vectorstate{\psi}\Pi)=\log |T_{Q_j}^n|\ge \log(\operatorname{poly}(n))+n\entropy(Q_j)$. Since $\typeclassprojection{n}{Q}\Pi=\typeclassprojection{n}{Q}$, we may estimate the norm of the projection as
\begin{equation}
\norm{\Pi\psi^{\otimes n}}^2\ge\norm{\typeclassprojection{n}{Q}\psi^{\otimes n}}^2=\lvert\typeclass{n}{Q}\rvert 2^{-n(\entropy(Q)+\relativeentropy{Q}{\mathcal{M}(\vectorstate{\psi})})}\polyge 2^{-n\relativeentropy{Q}{\mathcal{M}(\vectorstate{\psi})}}.
\end{equation}

It follows that
\begin{equation}
\begin{split}
E_{\alpha,\theta}(\psi^{\otimes n})
 & \ge \sum_{j=1}^k\theta(j)\entropy((\Pi\vectorstate{\psi}^{\otimes n}\Pi/\norm{\Pi\psi^{\otimes n}}^2)_j)+\frac{\alpha}{1-\alpha}\log\norm{\Pi\psi^{\otimes n}}^2  \\
 & \ge -\log\operatorname{poly}(n)+n\sum_{j=1}^k\theta(j)\entropy(Q_j)-n\frac{\alpha}{1-\alpha}\relativeentropy{Q}{\mathcal{M}(\vectorstate{\psi})}.
\end{split}
\end{equation}
For any $m\in\naturals$, $Q$ is also an $mn$-type, therefore
\begin{equation}
\frac{1}{mn}E_{\alpha,\theta}(\psi^{\otimes mn})\ge -\frac{\log\operatorname{poly}(mn)}{mn}+\sum_{j=1}^k\theta(j)\entropy(Q_j)-\frac{\alpha}{1-\alpha}\relativeentropy{Q}{\mathcal{M}(\vectorstate{\psi})}.
\end{equation}
By taking the limit $m\to\infty$, we obtain
\begin{equation}
E^{\alpha,\theta}(\psi)=\lim_{m\to\infty}\frac{1}{mn}E_{\alpha,\theta}(\psi^{\otimes mn})\ge\sum_{j=1}^k\theta(j)\entropy(Q_j)-\frac{\alpha}{1-\alpha}\relativeentropy{Q}{\mathcal{M}(\vectorstate{\psi})}.
\end{equation}
The inequality holds for every $n\in\naturals$ and $Q\in\distributions[n](\supp\mathcal{M}(\vectorstate{\psi}))$, so by density and continuity we replace the lower bound with its maximum over $Q\in\distributions(\supp\mathcal{M}(\vectorstate{\psi}))$, i.e., $E^{\alpha,\theta}(\psi)\ge\entropy_{\alpha,\theta}(\mathcal{M}(\vectorstate{\psi}))$. To sum up, we have the chain of inequalities
\begin{equation}
\rho^{\alpha,\theta}(\psi)
 \le \entropy_{\alpha,\theta}(\mathcal{M}(\vectorstate{\psi}))
 \le E^{\alpha,\theta}(\psi)
 \le \rho^{\alpha,\theta}(\psi),
\end{equation}
which implies that equality holds everywhere.
\end{proof}

\begin{example}\label{ex:W}
The \W{} state $\ket{W}=\frac{1}{\sqrt{3}}(\ket{100}+\ket{010}+\ket{001})$ has free support with respect to the computational basis. The three marginals have eigenvalues $(2/3,1/3)$, therefore at the vertices $\theta\in\{(1,0,0),(0,1,0),(0,0,1)\}$ we have $E^{\alpha,\theta}(\ket{W})=\rho^{\alpha,\theta}(\ket{W})=h_\alpha(1/3)$. For a general $\theta$ let us abbreviate the entries of a $Q\in\distributions(\supp\mathcal{M}(\vectorstate{\W}))$ as $q_1=Q(1,0,0)$, $q_2=Q(0,1,0)$, $q_3=Q(0,0,1)$. Then
\begin{equation}
\begin{split}
E^{\alpha,\theta}(\W)
 & = \entropy_{\alpha,\theta}(\mathcal{M}(\vectorstate{\W}))  \\
 & = \max_{q_1,q_2,q_3}\left[\theta(1)h(q_1)+\theta(2)h(q_2)+\theta(3)h(q_3)-\frac{\alpha}{1-\alpha}(\log 3-\entropy(q_1,q_2,q_3))\right],
\end{split}
\end{equation}
where the maximum is over triples $q_1,q_2,q_3\ge 0$ such that $q_1+q_2+q_3=1$.

At $\theta=(1/3,1/3,1/3)$, by concavity and invariance under permutations, the (unique) maximizer is the uniform distribution, and the maximum is $h(1/3)$ independently of $\alpha$.

At $\theta=(1/2,1/2,0)$, we can still simplify the problem by using concavity and invariance under exchanging $q_1$ and $q_2$, which implies that the maximum is attained at $q_1=q_2=q$, $q_3=1-2q$ for some $q\in[0,1/2]$, the solution to
\begin{equation}
\frac{\ed}{\ed q}\left[h(q)+\frac{\alpha}{1-\alpha}\entropy(q,q,1-2q)\right]=0
\end{equation}
While this equation cannot be solved for $q$, we observe that for each $q\in[1/3,1/2]$ it is possible to find the unique $\alpha\in[0,1]$ where the equality is satisfied, which gives the graph of $\alpha\to E^{\alpha,(1/2,1/2,0)}(\W)=\rho^{\alpha,(1/2,1/2,0)}(\W)$ in a parametric form (see \cref{fig:Wrholines}).
\end{example}
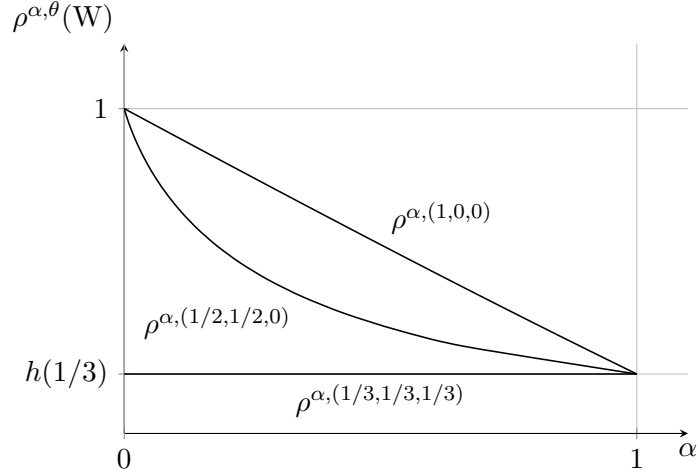
\begin{figure}
\begin{center}
\begin{tikzpicture}
\begin{axis}[spectrumplotstyle,
  ylabel=$\rho^{\alpha,\theta}(\W)$,
  ytick={0.918296, 1},
  yticklabels={$h(1/3)$, $1$},
  ymin=0.9,
  ymax=1.02,
  xtick={1},
]
\addplot[plotline,smooth] table [x=alpha,y=Wmax] {RhoWValues.dat} node[above right,pos=0.5] {$\rho^{\alpha,(1,0,0)}$};
\addplot[plotline, domain=0:1] {0.918296} node[below,pos=0.5] {$\rho^{\alpha,(1/3,1/3,1/3)}$};
\addplot[plotline,smooth] table [x=alpha,y=RhoThetaHalf] {WRhoThetaHalf.dat} node[below left,pos=0.65] {$\rho^{\alpha,(1/2,1/2,0)}$};
\end{axis}
\end{tikzpicture}
\end{center}
\caption{The values of $E^{\alpha,\theta}(\W)=\rho^{\alpha,\theta}(\W)$ as functions of $\alpha$ for three different values of $\theta$. As explained in \cref{ex:W}, $\rho^{\alpha,(1,0,0)}(\W)=h_\alpha(1/3)$ and $\rho^{\alpha,(1/3,1/3,1/3)}(\W)=h(1/3)$ independently of $\alpha$.}\label{fig:Wrholines}
\end{figure}

\begin{example}\label{ex:GHZ}
For a finite set $\mathcal{X}$ and a distribution $P\in\distributions(\mathcal{X})$, the generalized GHZ state is defined as
\begin{equation}
\ket{\GHZ_P}=\sum_{x\in\mathcal{X}}\sqrt{P(x)}\ket{x}\otimes\dots\otimes\ket{x}\in\complexes^\mathcal{X}\otimes\dots\otimes\complexes^\mathcal{X}.
\end{equation}
It is a direct sum of tensor product vectors, therefore the value of \emph{any} logarithmic spectral point of order $\alpha$ on this state is $\entropy_\alpha(P)$.

On the other hand, its support is free in the standard basis, therefore we can also use \cref{thm:freeErho} to evaluate $E^{\alpha,\theta}(\GHZ_P)$. The support of $\mathcal{M}(\vectorstate{\GHZ_P})$ is the diagonal $\setbuild{(x,x,\dots,x)}{x\in\mathcal{X}}$, therefore the marginals of any distribution $Q$ on the support may be identified with $Q$ itself. In particular, $\entropy(Q_j)=\entropy(Q)$ for all $j\in[k]$. This property leads to
\begin{equation}
\begin{split}
E^{\alpha,\theta}(\GHZ_P)
 & = \rho^{\alpha,\theta}(\GHZ_P)  \\
 & = \max_{Q \in \distributions(\supp\mathcal{M}(\GHZ))} \left[H(Q)-\frac{\alpha}{1-\alpha} \relativeentropy{Q}{\mathcal{M}(\GHZ)} \right]  \\
 & = \max_{Q \in \distributions(\mathcal{X})} \left[H(Q)-\frac{\alpha}{1-\alpha} \relativeentropy{Q}{P} \right]  \\
 & = H_\alpha(P),
\end{split}
\end{equation}
where the measurement operation $\mathcal{M}$ is understood in the computational basis and the last step uses \eqref{eq:variationalRenyi}. 

Consider now the special case $\mathcal{X}=\{0,1\}$, $P(0)=P(1)=\frac{1}{2}$. Choosing the $X$ basis $\ket{\pm}=\frac{1}{\sqrt{2}}(\ket{0}\pm\ket{1})$ as the local basis everywhere, the state vector can be expanded as
\begin{equation}
\begin{split}
\GHZ
 & = \frac{1}{\sqrt{2}}(\ket{00\dots 0}+\ket{11\dots 1})  \\
 & = 2^{-\frac{k-1}{2}}(\ket{++++\dots +}+\ket{--++\dots +}+\ket{-+-+\dots +}+\dots),
\end{split}
\end{equation}
where the sum is over all possible $\pm$ strings with an even number of $-$ symbols. The support with respect to this basis is therefore also free. Since the measured distribution $\mathcal{M}_{X}(\vectorstate{\GHZ})$ is uniform and so are its marginals, the optimal $Q$ is uniform as well, idependently of $\alpha$ and $\theta$ (although the maximizer is not unique if $\alpha=0$ and $k$ is even). This again leads to $E^{\alpha,\theta}(\GHZ)=1$, illustrating that the support of a state may be free with respect to different basis choices, but the value of $\entropy_{\alpha,\theta}$ is the same even if the measured distributions are very different.
\end{example}

\section{Transformation rates and examples}\label{sec:examples}

In this section we turn to the application of logarithmic spectral points and in particular the functionals $E^{\alpha,\theta}$ to strong converse exponents for asymptotic LOCC transformations. Recall that for transformations from the pure state $\psi$ to $\varphi$ the rate $R$ is achievable with error exponent $r\ge 0$ if there is a sequence of LOCC transformations such that with input $\psi^{\otimes n}$ the resulting state is $\varphi^{\otimes Rn +o(n)}$ with success probability at least $2^{-rn+o(n)}$. Let $R^*(\psi\to\varphi,r)$ denote the largest achievable rate. Rewriting the characterisation from \cite{jensen2019asymptotic} in terms of logarithmic functionals, we have
\begin{equation}
R^*(\psi\to\varphi,r)=\inf_{\alpha\in[0,1],E}\frac{r\frac{\alpha}{1-\alpha}+E(\psi)}{E(\varphi)},
\end{equation}
where the infimum is over all logarithmic spectral points of order $\alpha$ for all $\alpha\in[0,1]$.

Since we only know a subset of the asymptotic spectrum of LOCC transformations, in general we have the upper bound
\begin{equation}\label{eq:rateupperboundEalphatheta}
R^*(\psi\to\varphi,r)\le\inf_{\substack{\alpha\in[0,1]  \\  \theta\in\distributions([k])}}\frac{r\frac{\alpha}{1-\alpha}+E^{\alpha,\theta}(\psi)}{E^{\alpha,\theta}(\varphi)}.
\end{equation}
When the appearing states have free support, this upper bound can be computed by replacing $E^{\alpha,\theta}$ with $\rho^{\alpha,\theta}$ due to \cref{thm:freeErho}. For general states, we may use the single-letter upper and lower bounds $E_{\alpha,\theta}\le E^{\alpha,\theta}\le\rho^{\alpha,\theta}$ to obtain a potentially larger but easier-to-compute upper bound
\begin{equation}
R^*(\psi\to\varphi,r)\le\inf_{\substack{\alpha\in[0,1]  \\  \theta\in\distributions([k])}}\frac{r\frac{\alpha}{1-\alpha}+\rho^{\alpha,\theta}(\psi)}{E_{\alpha,\theta}(\varphi)}.
\end{equation}

Lower bounds on the rate, i.e., the achievability of certain rate-exponent pairs $(R,r)$ can be proved by exhibiting a protocol for single-shot transformation. In fact, such a protocol gives rise to a one-parameter family of achievable pairs (and any pair that can be obtained from these by decreasing $R$ or increasing $r$) as follows.
\begin{proposition}\label{prop:tradeofffromsingleshot}
Let $\psi,\varphi$ be (normalised) $k$-partite state vectors, $m,n\in\naturals$ and $p\in[0,1]$. Suppose that there exists an LOCC protocol that transforms $\psi^{\otimes n}$ into $\varphi^{\otimes m}$ with probability $p$. Then for every $q\in[0,1]$ the rate-exponent pair
\begin{equation}
(R,r)=\left(q\frac{m}{n},\frac{1}{n}\binaryrelativeentropy{\max\{p,q\}}{p}\right)
\end{equation}
is achievable. In the dual picture, for every $\alpha\in(0,1)$ and logarithmic spectral point $E$ of order $\alpha$ the inequality
\begin{equation}
E(\psi)\ge\frac{1}{n}\frac{\alpha}{1-\alpha}\log\left[p2^{m\frac{1-\alpha}{\alpha}E(\varphi)}+(1-p)\right]
\end{equation}
holds. For $\alpha=0$ this reduces to $E(\psi)\ge\frac{m}{n}E(\varphi)$, while for $\alpha=1$ we obtain $E(\psi)\ge p\frac{m}{n}E(\varphi)$.
\end{proposition}
\begin{proof}
Let $N\in\naturals$ and consider running the protocol $N$ times in parallel, on $Nn$ copies of $\psi$. The probability that at least $qN$ runs are successful is $\sum_{t=\lceil qN\rceil}^N\binom{N}{t}p^t(1-p)^{N-t}\polyeq 2^{-N\binaryrelativeentropy{\max\{p,q\}}{p}}$ (the asymptotic estimate follows from \cite[Lemma 2.6]{csiszar2011information} and that the number of terms is at most $N+1$), and in this case it is always possible to obtain $\lceil qN\rceil m$ copies of $\varphi$. This corresponds to the stated rate and error exponent.

For the second statement we note that the spectral point $F:=2^{(1-\alpha)E}$ is monotone also in the sense that if an LOCC protocol transforms $\ket{\psi}$ to the ensemble $(p_i,\ket{\psi_i})_{i\in I}$ then $F(\ket{\psi})^{1/\alpha}\ge\sum_{i\in I}p_iF(\ket{\psi_i})^{1/\alpha}$ (see \cite{jensen2019asymptotic} for details) and that in the worst case the given protocol results in a separable state with probability $1-p$.
\end{proof}
We will apply \cref{prop:tradeofffromsingleshot} to simple protocols in order to obtain a benchmark for the upper bound \eqref{eq:rateupperboundEalphatheta}.

\subsection{Transformations between a GHZ state and an arbitrary state}

Let $\psi$ be a state and consider transformations between $\psi$ and \GHZ{} states. The value of a logarithmic spectral point on $\GHZ_P$ only depends on the order, therefore the set of all pairs $(\alpha,E(\psi))$ contains all the information about the trade-off relations for transformations between $\psi$ and weighted \GHZ{} states (in fact, also between arbitrary weighted direct sums of tensor powers of $\psi$). For brevity we will refer to this set as the asymptotic spectrum of $\psi$ (more precisely, it is essentially the asymptotic spectrum of the subsemiring generated by $\psi$ and product vectors of arbitrary norm).

The precise determination of the set of pairs $(\alpha,E(\psi))$ is a difficult problem, which motivates studying inner and outer approximations, which translate to outer and inner approximations of the achievable rate region. Specifically, a point $(\alpha,E(\psi))$ in the asymptotic spectrum of $\psi$ certifies that any achievable pair $(r,R)$ for asymptotically transforming $\psi$ to \GHZ{} states must satisfy $R\le r\frac{\alpha}{1-\alpha}+E(\psi)$, i.e., excludes a half-plane from the set of rate-exponent pairs. Dually, if we show that a pair $(r,R)$ is achievable, then the region $\setbuild{(\alpha,e)}{e<R-r\frac{\alpha}{1-\alpha}}$ is not part of the asymptotic spectrum of $\psi$.

The image of the functionals $E^{\alpha,\theta}$ provides an inner bound of the asymptotic spectrum of $\psi$. By continuity in the parameters, for each $\alpha\in[0,1]$ the intersection with the corresponding coordinate line is a line segment, therefore we can completely describe the image by determining the minimum and the maximum over $\theta$ as functions of $\alpha$.

The special role of the minimum and the maximum in connection with transformations between $\psi$ and a \GHZ{} state can also be understood via the rate upper bound \eqref{eq:rateupperboundEalphatheta}. If the output is a generalised \GHZ{} state, it specialises to
\begin{equation}\label{eq:rateToGHZ}
\begin{split}
R^*(\psi\to\GHZ_P,r)
 & \le \inf_{\alpha,\theta} \frac{r\frac{\alpha}{1-\alpha}+E^{\alpha,\theta}(\ket{\psi})}{\entropy_\alpha(P)}  \\
 & = \inf_{\alpha\in[0,1]} \frac{r\frac{\alpha}{1-\alpha}+\min_{\theta\in\distributions([k])} E^{\alpha,\theta}(\ket{\psi})}{\entropy_\alpha(P)}.
\end{split}
\end{equation}
Conversely, if we transform a generalized \GHZ{} into an arbitrary state $\varphi$, we have 
\begin{equation}\label{eq:rateFromGHZ}
\begin{split}
R^*(\GHZ_P\to\varphi,r)
 & \le \inf_{\alpha,\theta} \frac{r\frac{\alpha}{1-\alpha}+\entropy_\alpha (P)}{E^{\alpha,\theta}(\ket{\varphi})}  \\
 & = \inf_{\alpha\in[0,1]} \frac{r\frac{\alpha}{1-\alpha}+\entropy_\alpha (P)}{\max_{\theta\in\distributions([k])} E^{\alpha,\theta}(\ket{\varphi})}.
\end{split}
\end{equation}
In both cases the optimisation over $\theta$ can be performed on $E^{\alpha,\theta}(\psi)$ separately, and the bound depends on the minimum (respectively maximum) over $\theta$ as a function of $\alpha$.

From \eqref{eq:logupperLOCC} we see that $\theta\mapsto E^{\alpha,\theta}(\psi)$ is the pointwise supremum of a family of affine functions, therefore it is convex. It follows that the maximum over $\theta$ is attained at an extreme point, i.e., one of the vertices of the simplex $\distributions([k])$. These are precisely the R\'enyi entanglement entropies between the corresponding subsystem and the rest of the system, i.e.,
\begin{equation}
\max_{\theta\in\distributions([k])}E^{\alpha,\theta}(\psi)=\max_{j\in[k]}\entropy(\Tr_j\vectorstate{\psi}).
\end{equation}
Thus \eqref{eq:rateFromGHZ} becomes
\begin{equation}
R^*(\GHZ_P\to\varphi,r) \le \inf_{\alpha\in[0,1]} \frac{r\frac{\alpha}{1-\alpha}+\entropy_\alpha (P)}{\max_{\theta\in\distributions([k])} E^{\alpha,\theta}(\ket{\varphi})},
\end{equation}
which is also what follows from bipartite bounds.

In contrast, \eqref{eq:rateToGHZ} in general improves on the bounds that we may obtain by considering the entanglement across bipartitions. If $\psi$ has free support, then $E^{\alpha,\theta}(\psi)=\rho^{\alpha,\theta}(\psi)$ is given by the maximum of a concave function, affine in the parameter $\theta$, therefore the computation of the minimal value for fixed $\alpha$ can be simplified using von Neumann's minimax theorem \cite{vneumann1928theorie}:
\begin{equation}\label{eq:minimaxFindMintheta}
\begin{split}
    \min_{\theta\in\distributions([k])}\entropy_{\alpha,\theta}(P)
    & = \min_{\theta\in\distributions([k])}\max_{Q\in\distributions(\supp P)}\left[\sum_{j=1}^k\theta(j)\entropy(Q_j)-\frac{\alpha}{1-\alpha}\relativeentropy{Q}{P}\right]\\
    & = \max_{Q\in\distributions(\supp P)}\min_{\theta\in\distributions([k])}\left[\sum_{j=1}^k\theta(j)\entropy(Q_j)-\frac{\alpha}{1-\alpha}\relativeentropy{Q}{P}\right]\\
    & = \max_{Q\in\distributions(\supp P)}\left[\min_j\entropy(Q_{j})-\frac{\alpha}{1-\alpha}\relativeentropy{Q}{P}\right].
\end{split}
\end{equation}

If the state $\psi$ is invariant under the permutations of the subsystems, then $E^{\alpha,\theta}(\psi)$ is also invariant under the permutations of $\theta$. Together with convexity, it implies that $\theta\mapsto E^{\alpha,\theta}(\psi)$ has a minumum at the uniform weight $\theta=(1/k,\dots,1/k)$, and a maximum at the Dirac weights $\theta=\delta_j$.

\begin{example}
We evaluate $\rho^{\alpha,\theta}$ for $\theta=(\frac{1}{3},\frac{1}{3},\frac{1}{3})$ on the normalised Coppersmith--Winograd tensor
\begin{equation}
    \CW_q=\frac{1}{\sqrt{3q}} \left(  \sum_{i=1}^q \ket{0ii} + \ket{i0i} + \ket{ii0}  \right).
\end{equation}
From \cite[Example 4.22.]{christandl2023universal} we know that the entropy term of $\rho^{\alpha,\theta}$ is maximized for the marginal probability distributions $(\frac{1}{3},\frac{2}{3}\frac{1}{q},\dots,\frac{2}{3}\frac{1}{q})$ for each subsystem, which originates from the uniform distribution over the support.  This distribution also minimizes the divergence because the elements of the \CW{} tensor in the defining basis also form the uniform distribution over the support. Therefore the term containing the divergence is zero, and we get the same value $\rho^{\alpha,\theta}(\CW_q)=\rho^{0,\theta}(\CW_q)=\rho^{\theta}(\CW_q)$, independently of $\alpha$:
\begin{equation}
\min_{\theta\in\distributions([k])}\rho^{\alpha,\theta}(\CW_q)
 = \rho^{(\frac{1}{3},\frac{1}{3},\frac{1}{3})}(\CW_q)
 = \frac{3}{2}\log q + h(\frac{1}{3})
\end{equation}
By the permutation symmetry, this is also the minimum of $\rho^{\alpha,\theta}(\CW_q)=E^{\alpha,\theta}(\CW_q)$ over $\theta$ for every $\alpha\in[0,1]$. The maximum is equal to the R\'enyi entropy of the marginal distribution:
\begin{equation}
\max_{\theta\in\distributions([k])}\rho^{\alpha,\theta}(\CW_q)
 = \rho^{\alpha,(1,0,0)}(\CW_q)
 = \frac{1}{1-\alpha}\log\left[\left(\frac{1}{3}\right)^{\alpha}+q\left(\frac{2}{3q}\right)^\alpha\right].
\end{equation}
\end{example}

If $\psi$ has free support but it is not permutation-invariant, then in general the minimising $\theta$ depends on $\alpha$ and one should use \eqref{eq:minimaxFindMintheta} directly to find the minimum. The expression inside the maximisation is the minimum of $k$ explicitly given concave functions, therefore it is also concave. The maximum can be efficiently computed numerically by convex optimization \cite{agrawal2018rewriting}. In the examples we used the Python library \texttt{cvxpy} \cite{diamond2016cvxpy} for computing $\rho^{\alpha,\theta}(\psi)$ and $\min_\theta\rho^{\alpha,\theta}(\psi)$.

\begin{example}
The state
\begin{equation}\label{eq:interestingExampleState}
\ket{\psi}=\sqrt{p}\ket{000}+\sqrt{\frac{1}{2}-p}\ket{011}+\sqrt{\frac{1}{2}-p}\ket{101}+\sqrt{p}\ket{112}
\end{equation}
has free support and for $p\in(0,p^*)$ where $p^*=0.113546\ldots$ is the nontrivial solution of $h(2x)+2x=1$, the third marginal has the highest max-entropy but the lowest von Neumann entropy. This implies that the maximum of $\rho^{\alpha,\theta}$ is attained at $\theta=(0,0,1)$ up to some $\alpha^*\in(0,1)$, while from that point it is attained at $\theta=(1,0,0)$ and $\theta=(0,1,0)$. This illustrates that the maximising weight may also depend on $\alpha$. \Cref{fig:InterestingExampleState} shows the inner bound on the asymptotic spectrum given by $E^{\alpha,\theta}(\psi)=\rho^{\alpha,\theta}(\psi)$ for $p=0.05$, with the minimum obtained numerically.
\end{example}
\begin{figure}
\begin{center}
\begin{tikzpicture}
\begin{axis}[spectrumplotstyle,
  ylabel=$\rho^{\alpha,\theta}(\psi)$,
  ytick={0.5689955935892812,1,1.584963},
  yticklabels={$h(2p)+2p$, $1$, $\log 3$},
  ymin=0,
  ymax=1.7,
  xtick={0,0.2,...,1},
  extra x ticks={0,0.46482},
  extra x tick labels={$0$,\raisebox{0pt}[3ex]{$\,\,\alpha^*$}},
  ]
\addplot[exampleline,name path=RhoTheta1] table [x=alpha,y=RhoTheta1] {Example.dat};
\addplot[exampleline,name path=RhoTheta3] table [x=alpha,y=RhoTheta3] {Example.dat};
\addplot[plotline,name path=MAX] table [x=alpha,y=RhoMax] {Example.dat};
\addplot[plotlinewithmarks,name path=MIN] table [x=alpha,y=RhoMin] {Example.dat};
\addplot[partofspectrumregion] fill between [of=MIN and MAX,soft clip={domain=0:1}];
\end{axis}
\end{tikzpicture}
\end{center}
\caption{Inner bound on the asymptotic spectrum of the state in \eqref{eq:interestingExampleState}. The lined region is the image of the functionals $\rho^{\alpha,\theta}$ evaluated at $\ket{\psi}=\sqrt{p}\ket{000}+\sqrt{\frac{1}{2}-p}\ket{011}+\sqrt{\frac{1}{2}-p}\ket{101}+\sqrt{p}\ket{112}$, with $p=0.05$. For each $\alpha\in[0,1]$ the image is an interval. The minimal values are found numerically, while the maximal value is $\max\{1,\entropy_\alpha(\Tr_3\vectorstate{\psi})\}$, attained at $\theta=(0,0,1)$ when $0\le\alpha\le\alpha^*= 0.46482\ldots$ and at $\theta=(1,0,0)$ and $\theta=(0,1,0)$ when $\alpha\ge\alpha^*$.}\label{fig:InterestingExampleState}
\end{figure}
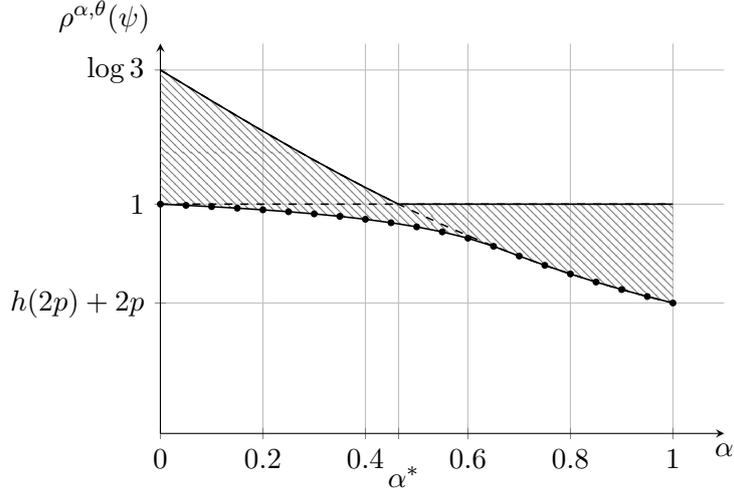

In general the function to be minimised in \eqref{eq:rateToGHZ} is not convex in $\alpha$. However, in the special case when the input state is free and the target is a uniform \GHZ{} state (i.e., the denominator is constant), the bound can be written as convex optimisation. To see this, let us introduce the new variable $t=\frac{\alpha}{1-\alpha}\in[0,\infty)$. Assuming without loss of generality that the target is a two-level \GHZ{} state, the upper bound becomes
\begin{equation}
\begin{split}
R^*(\psi\to\GHZ,r)
 & \le \inf_{\alpha,\theta}\left[r\frac{\alpha}{1-\alpha}+\rho^{\alpha,\theta}(\psi)\right]  \\
 & = \inf_{\alpha,\theta}\left[r\frac{\alpha}{1-\alpha}+\max_{P\in\distributions(\supp\mathcal{M}(\vectorstate{\psi}))}\left(\sum_{j=1}^k\theta(j)\entropy(P_j)-\frac{\alpha}{1-\alpha}\relativeentropy{P}{\mathcal{M}(\vectorstate{\psi})}\right)\right]  \\
 & = \inf_{t\in[0,\infty),\theta}\max_{P\in\distributions(\supp\mathcal{M}(\vectorstate{\psi}))}\left[rt+\sum_{j=1}^k\theta(j)\entropy(P_j)-t\relativeentropy{P}{\mathcal{M}(\vectorstate{\psi})}\right]  \\
 & = \max_{P\in\distributions(\supp\mathcal{M}(\vectorstate{\psi}))}\left[\inf_{t\in[0,\infty)}(r-\relativeentropy{P}{\mathcal{M}(\vectorstate{\psi})})t+\min_\theta\sum_{j=1}^k\theta(j)\entropy(P_j)\right]  \\
 & = \max_{\substack{P\in\distributions(\supp\mathcal{M}(\vectorstate{\psi}))  \\  \relativeentropy{P}{\mathcal{M}(\vectorstate{\psi})}\le r}}\min_j\entropy(P_j).
\end{split}
\end{equation}
In the second equality we have replaced $\alpha$ with $t$ so that the objective function is affine in the new variable (as well as in $\theta$), and concave in $P$, the third equality is due to Sion's minimax theorem \cite{sion1958general}, and the last equality uses that the infimum over $t$ is $0$ if $\relativeentropy{P}{\mathcal{M}(\vectorstate{\psi})}\le r$ (which is satisfied e.g. for $P=\mathcal{M}(\vectorstate{\psi})$) and $-\infty$ otherwise. We note that in the bipartite case the functionals $\rho^{\alpha,\theta}$ exhaust the asymptotic spectrum, therefore the inequality holds with equality. In this case the last line agrees with the rate formula given in \cite[Corollary 11]{hayashi2002error}.

We now turn to outer bounds on the asymptotic spectrum of $\psi$. A simple general upper bound on the maximum of $E(\psi)$ for a state vector $\psi\in\mathcal{H}=\mathcal{H}_1\otimes\dots\otimes\mathcal{H}_k$ is $\log\dim\mathcal{H}-\max_j\log\dim\mathcal{H}_j$. To see this, consider the following protocol, without loss of generality assuming that $\mathcal{H}_1$ has the largest dimension. With the notation $d_j=\dim\mathcal{H}_j$, a $d_2d_3\dots d_k$-level \GHZ{} state can be deterministically transformed into a product of $d_j$-dimensional maximally entangled pairs between subsystems $1$ and $j$ for all $j=2,3,\dots,k$. These can in turn be used to teleport the respective parts of a local copy of $\psi$, prepared at the first lab, to the remaining subsystems.

If $\psi$ is not separable across any bipartite cut, then we may find lower bounds on $E(\psi)$ in a similar way by studying transformations of (possibly several copies of) $\psi$ into \GHZ{} and determining achievable pairs.
\begin{example}\label{ex:WtoGHZ}
Consider the state $\ket{W}=\frac{1}{\sqrt{3}}(\ket{100}+\ket{010}+\ket{001})$. Since $\dim\mathcal{H}_A=\dim\mathcal{H}_B=\dim\mathcal{H}_C=2$, $E(\W)\le 3-1=2$ for any logarithmic spectral point $E$.

If either Alice, Bob, or Charlie performs a measurement in the computational basis, then with probability $2/3$ the result is $0$ and the post-measurement state is a maximally entangled pair between the other two subsystems. It follows that two \W{} states can be transformed into $\EPR_{AB}\otimes\EPR_{AC}$ with probability $4/9$, which in turn can be transformed deterministically (via teleportation) into a \GHZ{} state. By \cref{prop:tradeofffromsingleshot}, this implies that if $E$ is any logarithmic spectral point of order $\alpha$, then
\begin{equation}
E(\W)\ge\frac{1}{2}\frac{\alpha}{1-\alpha}\log\left[\frac{4}{9}2^{\frac{1-\alpha}{\alpha}}+\frac{5}{9}\right].
\end{equation}

On the other hand, from \cref{ex:W} we know that there exist logarithmic spectral points of order $\alpha$ which evaluate to $h(1/3)$ and $h_\alpha(1/3)$ on the \W{} state ($E^{\alpha,(1/3,1/3,1/3)}$ and $E^{\alpha,(1,0,0)}$, respectively). By continuity, $\setbuild{(\alpha,e)}{\alpha\in[0,1],e\in[h(1/3),h_\alpha(1/3)]}$ is part of the spectrum.

The inner and outer bounds on the asymptotic spectrum of the \W{} state and on the trade-off curve for transformations from \W{} to \GHZ{} are depicted in \cref{fig:WtoGHZrates,fig:RhoW}.    
\end{example}
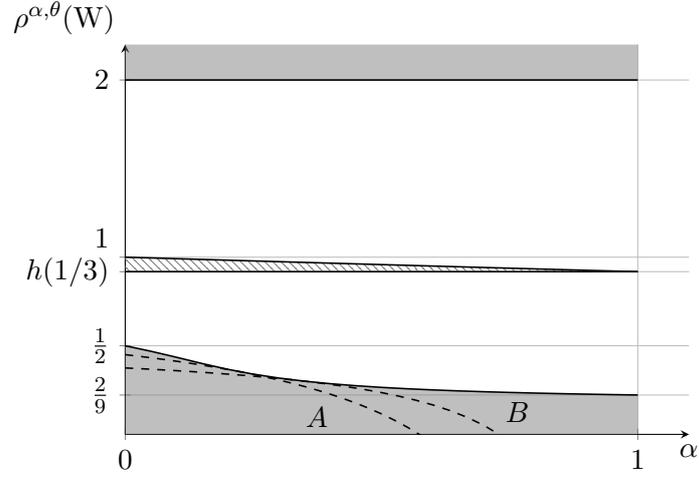
\begin{figure}
\begin{center}
\begin{tikzpicture}
\begin{axis}[spectrumplotstyle,
  ylabel=$\rho^{\alpha,\theta}(\W)$,
  ytick={0.222222, 0.5, 0.918296, 2},
  yticklabels={$\frac{2}{9}$, $\frac{1}{2}$, $h(1/3)$, $2$},
  ymin=0,
  ymax=2.2,
  xtick={1},
  extra y ticks={1},
  extra y tick labels={$1$},
  extra y tick style={
    grid=major,
    tick label style={anchor=south east}
  },
]
\path[name path=TOP] (0,\pgfkeysvalueof{/pgfplots/ymax})--(1,\pgfkeysvalueof{/pgfplots/ymax});
\path[name path=BOTTOM] (0,\pgfkeysvalueof{/pgfplots/ymin})--(1,\pgfkeysvalueof{/pgfplots/ymin});
\addplot[plotline,name path=MAX] table [x=alpha,y=Wmax] {RhoWValues.dat};
\addplot[plotline,name path=MIN, domain=0:1] {0.918296};
\addplot[plotline, name path=UBOUND] table [x=alpha,y=upperbound] {Wspecbounds.dat};
\addplot[plotline, name path=LBOUND] table [x=alpha,y=lowerbound] {Wspecbounds.dat};
\addplot[partofspectrumregion] fill between [of=MIN and MAX,soft clip={domain=0:1}];
\addplot[notpartofspectrumregion] fill between [of=UBOUND and TOP];
\addplot[notpartofspectrumregion] fill between [of=LBOUND and BOTTOM];
\addplot[exampleline, domain=0:0.6] {0.45-0.334368299182148*x/(1-x)} node[below left,pos=0.6] {$A$};
\addplot[exampleline, domain=0:0.8] {0.375-0.13908242663067*x/(1-x)} node[above right,pos=0.8] {$B$};
\end{axis}
\end{tikzpicture}
\end{center}
\caption{\label{fig:RhoW} Inner and outer bounds on the asymptotic spectrum of the \W{} state. The lined region represents the image of the functionals $\rho^{\alpha,\theta}$, which is a part of the (unknown) asymptotic spectrum. The dashed lines $A$ and $B$ are lower bounds corresponding to the similarly named points in \cref{fig:WtoGHZrates}.
}
\end{figure}

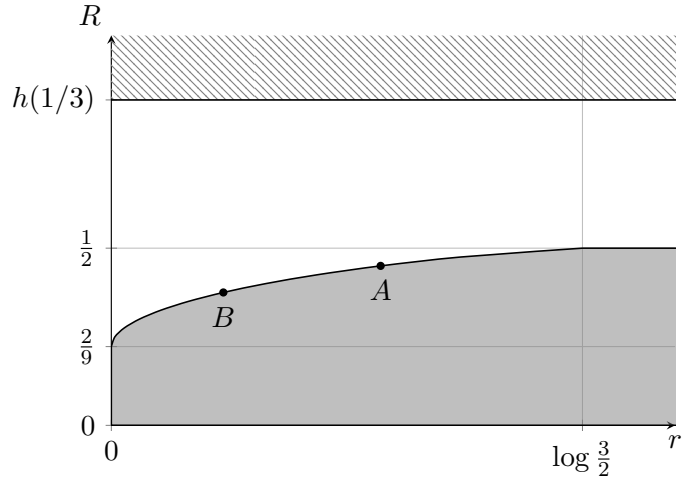
\begin{figure}
\begin{center}
\begin{tikzpicture}
\begin{axis}[
  rateplotstyle,
  ytick={0.5,0.9182958340544896},
  ytick={0.2222222222222222,0.5,0.9182958340544896},
  yticklabels={$\frac{2}{9}$,$\frac{1}{2}$,$h(1/3)$},
  ymin=0,
  ymax=1.1,
  extra y ticks={0},
  xtick={0,0.584962500721156},
  xticklabels={0,$\log\frac{3}{2}$},
]
\path[name path=TOP] (\pgfkeysvalueof{/pgfplots/xmin},\pgfkeysvalueof{/pgfplots/ymax})--(\pgfkeysvalueof{/pgfplots/xmax},\pgfkeysvalueof{/pgfplots/ymax});
\path[name path=BOTTOM] (\pgfkeysvalueof{/pgfplots/xmin},\pgfkeysvalueof{/pgfplots/ymin})--(\pgfkeysvalueof{/pgfplots/xmax},\pgfkeysvalueof{/pgfplots/ymin});
\addplot[plotline,name path=UBOUND] {0.9182958340544896};
\addplot[plotline,name path=LBOUND,smooth] table [x=r,y=R] {WtoGHZrates.dat} -- (1,0.5) \closedcycle;
\addplot[notachievableregion] fill between [of=UBOUND and TOP];
\addplot[achievableregion] fill between [of=LBOUND and BOTTOM];
\node[examplepoint] at (0.334368299182148,0.45) [label=below:$A$] {};
\node[examplepoint] at (0.13908242663067,0.375) [label=below:$B$] {};
\end{axis}
\end{tikzpicture}
\end{center}
\caption{\label{fig:WtoGHZrates} The rate-error exponent plane for transforming \W{} states into \GHZ{} states. The points in the lined region ($R>h(1/3)$) are not achievable. The shaded region, which corresponds to a protocol transforming two \W{} states into one \GHZ{} state with probability $4/9$ (see \cref{ex:WtoGHZ}), is part of the achievable region. The achievability of the points $A$ and $B$ imply outer bounds on the spectrum of the \W{} state (see \cref{fig:RhoW}).
}
\end{figure}

It is known that the $\alpha=0$ part of the asymptotic spectrum of the \W{} state is the interval $[h(1/3),1]$ \cite[Theorem 6.7 with $A=T^2$]{strassen1991degeneration}. While this part stays the same if we move within an SLOCC class, the $\alpha>0$ part does change in general, and so does the image of $E^{\alpha,\theta}$. For example, the weighted \W{} state
\begin{equation}\label{eq:Wex}
\ket{\W_{\textnormal{example}}}\coloneqq\sqrt{4/5}\ket{100}+\sqrt{1/10}\ket{010}+\sqrt{1/10}\ket{001}
\end{equation}
is in the \W{} class. The maximum is given by $\max_\theta E^{\alpha,\theta}(\W_{\textnormal{example}})=h_\alpha(4/5)$, and the minimum can be found numerically using \eqref{eq:minimaxFindMintheta} (\cref{fig:RhoWw}).
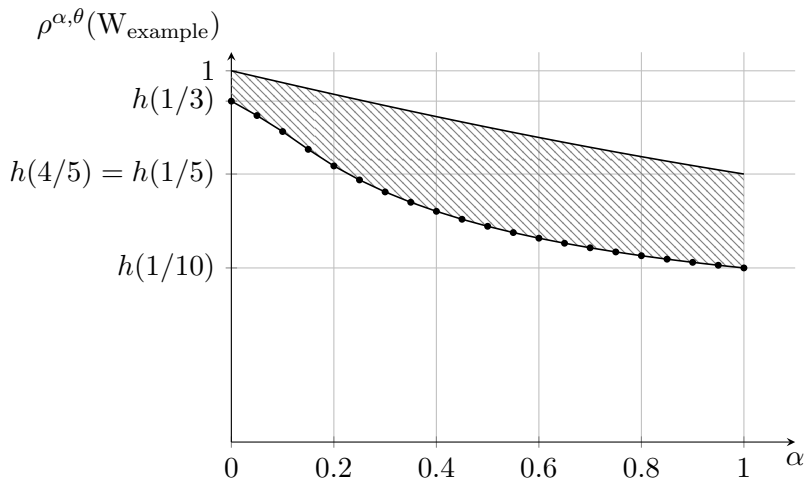
\begin{figure}
\begin{center}
\begin{tikzpicture}
\begin{axis}[spectrumplotstyle, ylabel=$\rho^{\alpha,\theta}(\W_{\textnormal{example}})$, ytick={0.468996,0.721928, 0.918296, 1}, yticklabels={$h(1/10)$, 
 $h(4/5)=h(1/5)$, $h(1/3)$, $1$}, ymin=0, ymax=1.05, xtick={0,0.2,...,1}]
\addplot[plotline,name path=MAX] table [x=alpha,y=Wwmax] {RhoWwValues.dat};
\addplot[plotlinewithmarks,name path=MIN] table [x=alpha,y=Wwmin] {RhoWwValues.dat};
\addplot[partofspectrumregion] fill between [of=MIN and MAX,soft clip={domain=0:1}];
\end{axis}
\end{tikzpicture}
\end{center}
\caption{ \label{fig:RhoWw} Inner bound on the asymptotic spectrum of a weighted \W{} state. The lined region is the image of the functionals $\rho^{\alpha,\theta}$ evaluated at $\ket{\W_{\textnormal{example}}}=\sqrt{4/5}\ket{100}+\sqrt{1/10}\ket{010}+\sqrt{1/10}\ket{001}$. For each $\alpha\in[0,1]$ the image is an interval. The minimal values are found numerically, while the maximal value is $h_\alpha(1/5)$, attained at $\theta=(1,0,0)$. For $\alpha=0$ the minimum is $h(1/3)$, attained at $\theta=(1/3,1/3,1/3)$, while for $\alpha=1$ it is $h(1/10)$, attained at $\theta=(0,x,1-x)$ for all $x\in[0,1]$.
}
\end{figure}

\subsection{Transforming W states to EPR pairs}

Recall that when $\psi$ and $\varphi$ both have free support, we have the rate bound
\begin{equation}
R^*(\psi\to\varphi,r)\le\inf_{\substack{\alpha\in[0,1]  \\  \theta\in\distributions([k])}}\frac{r\frac{\alpha}{1-\alpha}+\rho^{\alpha,\theta}(\psi)}{\rho^{\alpha,\theta}(\varphi)}.
\end{equation}
For each $\alpha,\theta$ the value can be found numerically by convex optimisation, but the expression inside the infimum is in general not convex. In practice we find that the infimum can be found using global optimization algorithms. In the example below we used the Python library \texttt{scipy.optimize} \cite{2020SciPy-NMeth} to find the upper bound on the strong converse rate.

As an example we choose transformations from (uniform) \W{} states to \EPR{} states shared between the first two subsystems (Alice and Bob). The state $\ket{\EPR_{AB}}=\frac{1}{\sqrt{2}}(\ket{000}+\ket{110})$ has free support with respect to the computational basis, and the measured distribution $\mathcal{M}(\ket{\EPR_{AB}})$ has uniform marginals, therefore this is where the maximum in $\entropy_{\alpha,\theta}(\mathcal{M}(\ket{\EPR_{AB}}))$ is attained. It follows that $E^{\alpha,\theta}(\EPR_{AB})=\rho^{\alpha,\theta}(\EPR_{AB})=\theta(1)+\theta(2)$, therefore the upper bound is
\begin{equation}\label{eq:WtoEPRupperbound}
R^*(\W\to\EPR_{AB},r)\le\inf_{\substack{\alpha\in[0,1]  \\  \theta\in\distributions([k])}}\frac{r\frac{\alpha}{1-\alpha}+\rho^{\alpha,\theta}(\W)}{\theta(1)+\theta(2)}.
\end{equation}
The result of the numerical optimisation is shown in \cref{fig:WtoEPRrates} (solid line with marks). We observe that for sufficiently large error exponents $r$, the upper bound is equal to $1$ within the precision of the numerical computation. Note that $1$ is the maximal rate for asymptotic SLOCC transformations, therefore it is an upper bound for every error exponent.

To find a lower bound, we employ \cref{prop:tradeofffromsingleshot} with the following simple protocol: Charlie performs a measurement in the computational basis, resulting in a maximally entangled state between Alice and Bob with probability $\frac{2}{3}$ and a separable state otherwise. It follows that for every $q\in[2/3,1]$ the pair
\begin{equation}\label{eq:WtoEPRlowerbound}
(R,r)=(q,\binaryrelativeentropy{q}{2/3})
\end{equation}
is achievable. In particular, $1$ is an achievable rate with strong converse exponent $r=\log\frac{3}{2}=0.58496\ldots$.

This raises the question whether the rate is strictly below $1$ when $r<\log\frac{3}{2}$, which is not settled by our numerical results. We show that this is indeed the case and in fact the upper bound \eqref{eq:WtoEPRupperbound} is strictly less than $1$. Based on the numerical computation we find that the optimal weights are $\theta=(1/2,1/2,0)$ for any $r$, therefore we start with the bound
\begin{equation}
R^*(\W\to\EPR_{AB},r)\le\inf_{\alpha\in[0,1]}\left(r\frac{\alpha}{1-\alpha}+\rho^{\alpha,(1/2,1/2,0)}(\W)\right).
\end{equation}
At $\alpha=1$ the expression inside the infimum is equal to $1$, therefore we can show that the infimum is strictly less than $1$ by showing that the right derivative (as a function of $\alpha$) is negative. This can be seen by \cref{cor:difHalphathetaByalphaAtZero} with $P=(1/3,1/3,1/3)$ and $Q_0=(1/2,1/2,0)$ as follows:
\begin{equation}
\left.\frac{\ed}{\ed\alpha}\left[r\frac{\alpha}{1-\alpha}+\rho^{\alpha,(1/2,1/2,0)}(\W)\right]\right|_{\alpha=0}
 = r-\relativeentropy{Q_0}{P} 
 = r-\log\frac{3}{2}.
\end{equation}
This means that if $r<\log 3/2$, the infimum cannot be attained at $\alpha=0$. It also follows that tangent line of the upper bound curve on the $(r,R)$ plane corresponding to any $r<\log 3/2$ has strictly positive slope.
\begin{figure}
\begin{center}
\begin{tikzpicture}
\begin{axis}[
  rateplotstyle,
  ytick={0.666666666666,0.9182958340544896,1},
  yticklabels={$\frac{2}{3}$,$h(1/3)$,$1$},
  ymin=0.6,
  ymax=1.1,
  xtick={0,0.584962500721156},
  xticklabels={0,$\log\frac{3}{2}$},
  extra y ticks = {0.6},
]
\path[name path=TOP] (\pgfkeysvalueof{/pgfplots/xmin},\pgfkeysvalueof{/pgfplots/ymax})--(\pgfkeysvalueof{/pgfplots/xmax},\pgfkeysvalueof{/pgfplots/ymax});
\path[name path=BOTTOM] (\pgfkeysvalueof{/pgfplots/xmin},\pgfkeysvalueof{/pgfplots/ymin})--(\pgfkeysvalueof{/pgfplots/xmax},\pgfkeysvalueof{/pgfplots/ymin});
\addplot[plotlinewithmarks,smooth,name path=UBOUND] table [x=r,y=R] {rAndRW.dat} -- (1,1);
\addplot[notachievableregion] fill between [of=UBOUND and TOP];
\addplot[exampleline, domain=0:0.8] {x + 0.9333287745328493};
\addplot[exampleline, domain=0:0.8] {0.052631578947368425*x + 0.981607614033931};
\addplot[plotline,name path=LBOUND,domain=0.666666666666:1, variable=q, samples=50] ({q*ln(q*3/2)/ln(2)+(1-q)*ln((1-q)*3)/ln(2)}, q) -- (1,1) \closedcycle;
\addplot[achievableregion] fill between [of=LBOUND and BOTTOM];
\end{axis}
\end{tikzpicture}
\end{center}
\caption{\label{fig:WtoEPRrates} The rate-error exponent plane for transforming \W{} states into \EPR{} pairs between Alice and Bob. The points in the lined region are not achievable. The boundary is found by numerically evaluating the rate upper bound \eqref{eq:WtoEPRupperbound}. The shaded region, is part of the achievable region (see \eqref{eq:WtoEPRlowerbound}). The dashed lines show the bounds provided by a single $\rho^{\alpha,\theta}$ with $\alpha=0.5$ and $\alpha=0.05$.
}
\end{figure}
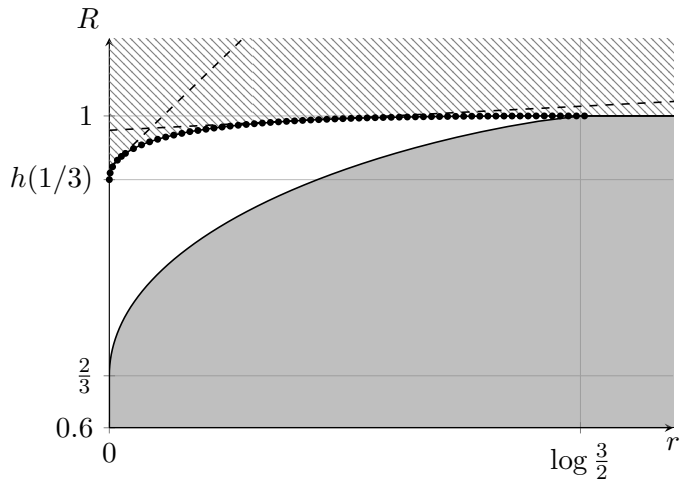

\section{Conclusion}

The functionals $\rho^{\alpha,\theta}$ introduced in this paper and $E^{\alpha,\theta}$, $E_{\alpha,\theta}$ are analogous and also related in the limit $\alpha\to 0$ to the upper support functionals $\rho^{\theta}$ and the quantum functionals $E^\theta$, $E_\theta$. For instance, $\rho^{\alpha,\theta}$ and $\rho^{\theta}$ are subadditive and $E^{\alpha,\theta}$, $E^{\theta}$ are (fully) additive under tensor products, and the inequalties $E^{\theta}\le\rho^{\theta}$ and $E^{\alpha,\theta}\le\rho^{\alpha,\theta}$ hold. Moreover, the equalities $\lim_{n\to\infty}\frac{1}{n}\rho^\theta(\psi^{\otimes n})=E^\theta(\psi)$ and $\lim_{n\to\infty}\frac{1}{n}\rho^{\alpha,\theta}(\psi^{\otimes n})=E^{\alpha,\theta}(\psi)$ hold.

However, the analogy is not complete. In \cite{christandl2023universal} is was shown that $E_\theta=E^\theta$, while we only know that $E_{\alpha,\theta}\le E^{\alpha,\theta}$ and that $E_{\alpha,\theta}$ is superadditive and its regularisation is $E^{\alpha,\theta}$ \cite{vrana2023family}. Also, while $\lim_{\alpha\to 0}E^{\alpha,\theta}=E^\theta$ and $\lim_{\alpha\to 0}E_{\alpha,\theta}=E_\theta$, the limit $\lim_{\alpha\to 0}\rho^{\alpha,\theta}$ is probably not equal to $\rho^{\theta}$. The reason for this possible discrepancy is the difference between minimising over product orthonormal bases and all product bases, which may result in the limit being strictly greater than $\rho^\theta$. Another difference is that while $\rho^\theta$, $E^\theta$, $E_\theta$ (and also $\rho_\theta$) are known to be monotone decreasing under SLOCC transformations, we do not know whether $\rho^{\alpha,\theta}$ and $E_{\alpha,\theta}$ are monotone under LOCC ($E^{\alpha,\theta}$ is a spectral point, therefore it is monotone).

It is conceivable that the limit \eqref{eq:Egeneralbound} can also be evaluated when $T$ is an arbitrary quantum-classical channel (essentially a POVM) instead of a von Neumann measurement. This would lead to a subadditive functional between $E^{\alpha,\theta}$ and $\rho^{\alpha,\theta}$, with regularisation equal to $E^{\alpha,\theta}$, and which potentially recovers $\rho^\theta$ in the limit $\alpha\to 0$. It would be interesting to see if, with the help of this functional, one could explicitly find the value of $E^{\alpha,\theta}$ on a larger set of states than what is considered in \cref{thm:freeErho}.

Finally, we mention that the functionals $\rho^{\theta}$ have a lower counterpart $\rho_{\theta}$ as well \cite[Section 3.]{strassen1991degeneration}, which is superadditive and satisfies $\rho_\theta\le E_\theta=E^\theta\le\rho^\theta$, but we do not have a candidate for a possible generalisation to $\alpha>0$. 

\section*{Acknowledgement}

Supported by the \'UNKP-23-5-BME-458 New National Excellence Program of the Ministry for Culture and Innovation from the source of the National Research, Development and Innovation Fund, the J\'anos Bolyai Research Scholarship of the Hungarian Academy of Sciences, and by by the Ministry of Culture and Innovation and the National Research, Development and Innovation Office within the Quantum Information National Laboratory of Hungary (Grant No.~2022-2.1.1-NL-2022-00004).

\bibliography{references}{}

\end{document}